\documentclass{CSML}

\def\dOi{12(3:6)2016}
\lmcsheading%
{\dOi}
{1--41}
{}
{}
{Nov.~30, 2015}
{Sep.~\phantom05, 2016}
{}

\ACMCCS{[{\bf Theory of computation}]: Logic---Logic and verification;
  [{\bf Software and its engineering}]:  Software organization and properties---Software functional properties---Formal methods---Software verification}

\usepackage{amssymb,amsthm}
\usepackage{stmaryrd} 
\usepackage{hyperref}
\usepackage{color}

\usepackage{bussproofs}
\newcommand{\TernaryInfC}{\TrinaryInfC}

\usepackage{pgf,tikz}
\usetikzlibrary{arrows,automata}
\usetikzlibrary{shapes}

\theoremstyle{plain}

\newtheorem{theorem}[thm]{Theorem}
\newtheorem{corollary}[thm]{Corollary}
\theoremstyle{definition}
\newtheorem{lemma}[thm]{Lemma}

\newtheorem{example}[thm]{Example}

\newtheorem{note}[thm]{Note}
\newtheorem{observation}[thm]{Observation}
\newtheorem{claim}[thm]{Claim}
\newtheorem*{claim*}{Claim}

\newcommand{\tsum}{\textstyle\sum}
\newcommand{\tbigcup}{\textstyle\bigcup}
\newcommand{\tbigcap}{\textstyle\bigcap}
\newcommand{\tbigvee}{\textstyle\bigvee}
\newcommand{\tbigwedge}{\textstyle\bigwedge}

\newcommand{\tto}{\rightarrowtriangle}
\newcommand{\lrto}{\leftrightarrow}
\newcommand{\lt}{\langle}
\newcommand{\rt}{\rangle}
\renewcommand{\star}{^*}

\newcommand{\llt}{\lt\!\lt}
\newcommand{\rrt}{\rt\!\rt}

\newcommand{\play}{\mathsf{play}}

\newcommand{\hoare}[3]{\{#1\}#2\{#3\}}
\newcommand{\Imp}{\Rightarrow}

\newcommand{\ang}{\sqcup}
\newcommand{\bigAng}{\bigsqcup}
\newcommand{\tbigAng}{\textstyle\bigAng}
\newcommand{\dem}{\sqcap}
\newcommand{\bigDem}{\bigsqcap}
\newcommand{\tbigDem}{\textstyle\bigDem}
\newcommand{\true}{\mathsf{true}}
\newcommand{\false}{\mathsf{false}}
\newcommand{\id}{\mathsf{id}}
\newcommand{\compl}{{\sim}}
\newcommand{\tbigplus}{\tsum}

\newcommand{\kwIf}{\mathsf{if}}
\newcommand{\kwThen}{\mathsf{then}}
\newcommand{\kwWhile}{\mathsf{while}}
\newcommand{\kwDo}{\mathsf{do}}

\newcommand{\kwElse}{\mathsf{else}}
\newcommand{\ifThenElse}[3]{\kwIf\, #1 \,\kwThen\, #2 \,\kwElse\, #3}

\newcommand{\whileDo}[2]{\kwWhile\, #1 \,\kwDo\, #2}
\newcommand{\wh}{\textsf{\large w}}

\newcommand{\kc}{;}
\newcommand{\plus}{\boldsymbol +}
\newcommand{\lb}{\boldsymbol [}
\newcommand{\rb}{\boldsymbol ]}
\newcommand{\skwWh}{\mathrm{w\hspace{-0.5pt}h}}
\newcommand{\skwDo}{\mathrm{d\hspace{-0.5pt}o}}
\newcommand{\sWhileDo}[2]{\skwWh\, #1 \,\skwDo\, #2}

\newcommand{\nto}{\leadsto}

\usepackage[outline]{contour}
\newcommand{\gto}{\mathbin{\textcolor{white}{\contour{black}{$\nto$}}}}
\contourlength{0.05em}
\newcommand{\gc}{\boldsymbol ;\,}
\newcommand{\glb}{\scalebox{1.5}[1.1]{$\boldsymbol [$}}
\newcommand{\grb}{\scalebox{1.5}[1.1]{$\boldsymbol ]$}}
\newcommand{\gang}{\boldsymbol\ang}
\newcommand{\gdem}{\boldsymbol\dem}
\newcommand{\zero}{\mathbf{0}}
\newcommand{\one}{\mathbf{1}}
\newcommand{\gkwWh}{\mathbf{w\hspace{-0.5pt}h}}
\newcommand{\gkwDo}{\mathbf{d\hspace{-0.5pt}o}}
\newcommand{\gWhileDo}[2]{\gkwWh\, #1 \,\gkwDo\, #2}

\newcommand{\lift}{\textstyle\mathop{\mathsf{lift}}}

\newcommand{\ord}{\mathsf{ord}}
\newcommand{\Ord}{\mathbf{Ord}}

\newcommand{\At}{\mathsf{At}}
\newcommand{\AtC}{{\mathsf{At}_\Phi}}

\newcommand{\rHyp}{\mathsf{hyp}}
\newcommand{\rSkip}{\mathsf{skip}}
\newcommand{\rDvrg}{\mathsf{dvrg}}

\newcommand{\rSeq}{\mathsf{seq}}
\newcommand{\rCond}{\mathsf{cond}}
\newcommand{\rLoop}{\mathsf{loop}}

\newcommand{\rWeak}{\mathsf{weak}}

\newcommand{\rAng}{\mathsf{ang}}
\newcommand{\rDem}{\mathsf{dem}}
\newcommand{\rJoin}{\mathsf{join}}
\newcommand{\rMeet}{\mathsf{meet}}

\newcommand{\gI}{I_{\Phi\Psi}}
\newcommand{\dI}{J_{\Phi\Psi}}

\newcommand{\class}[1]{\mathcal{#1}}
\newcommand{\All}{\mathit{All}}
\newcommand{\Dem}{\mathit{Dem}}

\newcommand{\cl}{C}

\newcommand{\simpl}{\boldsymbol\sqsubseteq}

\newcommand{\EXPTIME}{\mathsf{EXPTIME}}
\newcommand{\APSPACE}{\mathsf{APSPACE}}
\newcommand{\B}{\textvisiblespace}
\newcommand{\halt}{\mathit{halt}}
\newcommand{\accept}{\mathit{accept}}
\newcommand{\reject}{\mathit{reject}}
\newcommand{\writeAt}[1]{\mathsf{write}\ #1}
\newcommand{\move}[1]{\mathsf{move}\ #1}
\newcommand{\switch}[1]{\mathsf{switch}\ #1}

\begin{document}

\title[Synthesis of Strategies Using the Hoare Logic of Dual Nondeterminism]{Synthesis of Strategies Using the Hoare Logic of Angelic and Demonic Nondeterminism\rsuper*}

\author[K.~Mamouras]{Konstantinos Mamouras}
\address{Department of Computer and Information Science, University of Pennsylvania, Philadelphia, PA}
\email{mamouras@seas.upenn.edu}



\titlecomment{{\lsuper*}This is a revised and expanded version of the paper \cite{mamouras2015angHL}, which was presented in FoSSaCS 2015.}


\begin{abstract}
We study a propositional variant of Hoare logic that can be used for reasoning about programs that exhibit both angelic and demonic nondeterminism. We work in an uninterpreted setting, where the meaning of the atomic actions is specified axiomatically using hypotheses of a certain form. Our logical formalism is entirely compositional and it subsumes the non-compositional formalism of safety games on finite graphs. We present sound and complete Hoare-style calculi that are useful for establishing partial-correctness assertions, as well as for synthesizing implementations. The computational complexity of the Hoare theory of dual nondeterminism is investigated using operational models, and it is shown that the theory is complete for exponential time.
\end{abstract}

\maketitle

\section{Introduction}
\label{sec:intro}

Demonic nondeterminism is used in the context of programming to model external influences which are not under the control of the program. Such nondeterminism may arise in concurrent programs, for example, from the scheduling of threads, which is under the control of the operating system and not the program. Others examples could be sensor readings or user input, which are completely external influences to a computing system. In the case of user input, in particular, we can typically make no assumptions, since the input depends on an entirely unpredictable and uncontrollable human being, who may choose to behave as an adversary.

Even in the absence of ``real'' nondeterminacy like scheduling and sensor/user input, we may use demonic nondeterminism to represent abstraction and partial knowledge of the state of a computation. An example of the latter use of demonic nondeterminism is when we cannot fully observe the value of an integer variable $x$, but we can tell whether it is negative, zero, or positive. At this level of abstraction, we cannot describe the operation $x := x+1$ that increments the variable $x$ by 1 deterministically.
\begin{gather*}
\begin{gathered}[t]
\textbf{Observe $x$ fully} \\
\text{deterministic action} \\[-0.5ex]
x \mapsto x+1
\end{gathered}
\qquad
\begin{gathered}[t]
\\
\xrightarrow{\text{\normalsize abstraction}}{}
\end{gathered}
\quad
\begin{gathered}[t]
\textbf{Observe $x$ partially} \\
\text{corresponding nondeterministic action} \\[-0.5ex]
\begin{aligned}[t]
(x<0) &\mapsto (x<0) \lor (x=0) \\[-0.5ex]
(x=0) &\mapsto (x>0) \ \text{and}\ 
(x>0) \mapsto (x>0)
\end{aligned}
\end{gathered}
\end{gather*}
This example illustrates that nondeterminism is necessary when creating finite-state abstractions of realistic programs, whose state space is typically infinite.

Angelic nondeterminism, on the other hand, is used to express nondeterminacy that is under the control of the program. We use angelic nondeterminism to leave some implementation details of a program underspecified. The ``angel'', namely the agent that represents our interests, controls how these details are resolved in order to achieve the desired result. The process of resolving these implementation details amounts to \emph{synthesizing} a fully specified program. The term \emph{dual nondeterminism} refers to the combination of angelic and demonic nondeterminism.

In order to reason about dual nondeterminism, one first needs to have a semantic model of how programs with angelic and demonic choices compute. One semantic model that has been used extensively uses a class of mathematical objects that are called monotonic predicate transformers \cite{back1998} (based on Dijkstra's predicate transformer semantics \cite{dijkstra1975, morgan1998}). An equivalent denotational model that is based on binary relations was introduced in \cite{rewitzky2003} (up-closed multirelations) and further investigated in \cite{martin2004, martin2007, martin2013}. These relations can be understood intuitively as two-round games between the angel and the demon.

\newcommand{\even}{\mathit{even}}
\newcommand{\odd}{\mathit{odd}}

We are interested here in verifying properties of programs that can be expressed as Hoare (partial-correctness) assertions \cite{floyd1967, hoare1969, cook1978, apt1981HL, apt1983HL}, that is, formulas of the form $\hoare{p}f{q}$, where $f$ is the program text and $p, q$ denote predicates on the state space, called precondition and postcondition respectively. The formula $\hoare{p}f{q}$ asserts, informally, that starting from any state satisfying the precondition $p$, the angel has a strategy so that whatever the demon does, the final state of the computation of $f$ (assuming termination) satisfies the postcondition $q$. This describes a notion of partial correctness, because in the case of divergence (non-termination) the angel wins vacuously. Our language for programs and preconditions/postconditions involves abstract test symbols $p, q, r, \ldots$ and abstract action symbols $a, b, \ldots$ with no fixed interpretation. We constrain their meaning with extra hypotheses: we consider a finite set $\Phi$ of Boolean axioms for the tests, and a finite set $\Psi$ of axioms of the form $\hoare{p}a{q}$ for the action letters. So, we typically assert implications of the form
\[
  \Phi,\Psi \Imp \hoare{p}f{q},
\]
which we call \emph{simple Hoare implications}.
For example, consider the tests $\even(n)$, $\odd(n)$ and the action $n{+}{+}$, which increments $n$ by 1. We think that these are abstract symbols contrained by the hypotheses $\Phi$ and $\Psi$ below.
\begin{gather*}
\begin{aligned}[t]
\Phi : {}
& \even(n) \lor \odd(n)
\\
& \neg \even(n) \lor \neg \odd(n)
\end{aligned}
\qquad
\begin{aligned}[t]
\Psi :
\hoare{\even(n)}{&n{+}{+}}{\odd(n)}
\\
\hoare{\odd(n)}{&n{+}{+}}{\even(n)}
\end{aligned}
\qquad
\begin{aligned}[t]
f := {}
&\kwIf\ \even(n)\ \kwThen\ n{+}{+}
\\
&\kwElse\ n{+}{+};n{+}{+}
\end{aligned}
\end{gather*}
We should be able to prove that $\Phi,\Psi \Imp \hoare{\true}f{\odd(n)}$ under the above definitions.
We want to design a formal system that allows the derivation of the valid Hoare implications. One important desideratum for such a formal system is to also provide us with program text that corresponds to the winning strategy of the angel. Then, the system can be used for the deductive synthesis of programs that satisfy their Hoare specifications.

There has been previous work on deductive methods to reduce angelic nondeterminism and synthesize winning strategies for the angel. The work \cite{celiku2003}, which is based on ideas of the refinement calculus \cite{back1990dual, back1992ADM, back1998, morgan1998}, explores a total-correctness Hoare-style calculus to reason about angelic nondeterminism.
It is observed that there is a conceptual difficulty in reconciling \emph{nondeterministic refinement} (which results from removing demonic choices or/and adding angelic choices) with the task of synthesizing the strategy of the angel. This is because the interaction between the angel and the demon has been fixed in advance: we have no control over the demonic nondeterminism, and increasing the choices of the angel is not permitted. Nonetheless, a refinement-based approach for implementing angelic choices is pursued in \cite{celiku2003}.
The analysis is in the first-order interpreted setting, and no completeness or relative completeness results are discussed.

Of particular relevance to our investigations is the line of work that concerns two-player infinite games played on finite graphs \cite{thomas1995}. Such games are useful for analyzing (nonterminating) reactive programs. One of the players represents the ``environment'', and the other player is the ``controller''. Computing the strategies that witness the winning regions of the two players amounts to synthesizing an appropriate implementation for the controller. The formalism of games on finite graphs is very convenient for developing an algorithmic theory of synthesis. However, the formalism is non-succinct and, additionally, it is inherently non-compositional. An important class of properties for these graph games are the so called \emph{safety} properties, which assert that the environment cannot force the play into a ``bad'' region. For encoding safety properties, we see that a fully compositional formalism based on while programs and partial-correctness properties suffices.

\paragraph{\bfseries Our Contribution} We consider a propositionally abstracted language for while programs with demonic and angelic choices. Our results are the following:
\begin{itemize}[label=$-$]
\item
We give the intended operational semantics in terms of safety games on graphs, and we describe a denotational semantics based on a restricted subclass of multirelations. We obtain a full abstraction result for all reasonable intepretations of the atomic symbols, which asserts the equivalence between the operational and denotational models.
\item
We present a sound and \emph{unconditionally} complete calculus for the weak Hoare theory of dual nondeterminism (over the class of all interpretations). We also consider a restricted class of interpretations, where the atomic actions are non-angelic, and we extend our calculus so that it is complete for the Hoare theory of this smaller class (called strong Hoare theory). The proofs of these results rely on the construction of free models.
\item
Using the correspondence between the operational and denotational models, we prove that the strong Hoare theory of dual nondeterminism is $\EXPTIME$-complete.
\item
We consider an extension of our Hoare-style calculus with annotations that denote the winning strategies of the angel. We thus obtain a sound and complete deductive system for the synthesis of angelic strategies.
\item
Our formalism is shown to subsume that of safety games on finite graphs, hence it provides a compositional method for reasoning about safety in reactive systems. The language of dually nondeterministic program schemes is exponentially more succinct than explicitly represented game graphs, and it is arguably a more natural language for describing algorithms and protocols.
\end{itemize}
The present paper is a revised and extended version of \cite{mamouras2015angHL}. We include here all the proofs that were omitted from the conference version \cite{mamouras2015angHL}, and we generalize the full abstraction result on the correspondence between the operational and denotational semantics. In \cite{mamouras2015angHL}, full abstraction was established only for the free models, which are finite. In order to generalize the full abstraction theorem to infinite models, we identify here a natural condition on the interpretations of atomic actions (which we call \emph{chain property}). This condition covers all finite models, as well as all infinite models with a ``reasonable'' interpretation of the atomic actions.

\paragraph{\bfseries Outline of paper}

In \S\ref{sec:prelim} we recall some well-known definitions and facts about abstract imperative while programs, and we introduce the relevant notation that we will use in our later development. We introduce \emph{while game schemes} in \S\ref{sec:operational}, which are abstractions of programs that allow both angelic and demonic nondeterministic choices. We also present in \S\ref{sec:operational} the \emph{intended operational semantics}, which is based on the familiar model of two-player safety games played on graphs. We explore in \S\ref{sec:denotational} a denotational model based on a certain kind of binary relations. We show that this denotational semantics extends naturally the standard relational semantics of programs, and additionally it agrees exactly with the intended operational model. In \S\ref{sec:hoare} we introduce the syntax and meaning of Hoare assertions and implications, and we propose a Hoare-style calculus for reasoning about while game schemes. Our first completeness result is given in \S\ref{sec:completeA}, where we show that the partial-correctness calculus of \S\ref{sec:hoare} is complete for the \emph{weak Hoare theory} (the theory over the class of all interpretations). In \S\ref{sec:completeB} we study the \emph{strong Hoare theory}, which is the theory over the subclass of interpretations that assign a non-angelic meaning to the atomic actions. We extend our calculus to completeness for this important case, and we show that the theory is complete for $\EXPTIME$. We further extend in \S\ref{sec:synthesis} our axiomatization of the strong Hoare theory with annotations that witness the angelic strategies. We thus obtain a sound and complete Hoare-style calculus for the synthesis of angelic implementations. It is also shown that our formalism subsumes the (non-compositional and non-succinct) formalism of safety games on finite graphs. We analyze a simple example in \S\ref{sec:example} for a toy temperature controller, which illustrates in a very concrete way how our verification/synthesis calculus can be used. In \S\ref{sec:related} we discuss several related works, including the ones from which the present paper was inspired. We conclude in \S\ref{sec:conclusion} with a brief summary of our technical contribution, and with suggestions for future work.

\section{Preliminaries: Monadic While Program Schemes}
\label{sec:prelim}

In this section we give some preliminary definitions regarding abstract imperative programs with while loops, which are also known in the literature as \emph{while program schemes}. See for example \cite{rutledge1964IPS, paterson1968PS, luckham1970FCP, paterson1970CS, garland1973PS} for some very well-known works in the area of program schematology. The programs that we consider here are often qualified as \emph{monadic}, which means that the program state is considered to be one indivisible entity. In other words, the program actions are modeled as unary functions that act on the entire program state. There are no distinct program variables $x,y,z,\ldots$ at the syntactic level, nor variable assignments $z \gets f(x,y)$ that can read from and assign to variables individually. Instead, the primitive actions are written simply as atomic letters $a, b, c, \ldots$ that should be thought as transforming the whole program state. Alternatively, one can think equivalently that there is a single program variable $x$ (which represents the entire program state) and an atomic action $a$ corresponds to an assignment $x \gets a(x)$.

We are interested in program schemes that allow the use of the construct $\dem$ of demonic nondeterministic choice. This is a very useful operation, because it can model underspecification and real nondeterminism (environment, user input, and so on). First, we present the syntax of these abstract while programs. Then, we give the standard denotational semantics for them, which is based on binary relations.

\begin{defi}[The Syntax of Program Schemes]
\label{def:syntaxWPS}
We consider a two-sorted algebraic language. There is the sort of \emph{tests} and the sort of \emph{programs}. The tests are built up from \emph{atomic tests} and the constants $\true$ and $\false$, using the usual Boolean operations: $\neg$ (negation), $\land$ (conjunction), and $\lor$ (disjunction). We use the letters $p, q, r, \ldots$ to range over arbitrary tests.
Tests are thus given by the grammar:
\[
  \text{tests $p, q$} ::=
  \text{atomic test} \mid
  \true \mid
  \false \mid
  \neg p \mid
  p \land q \mid
  p \lor q.
\]
As usual, the implication $p \to q$ is abbreviation for $\neg p \lor q$, and the double implication $p \lrto q$ stands for $(p \to q) \land (q \to p)$.

The base programs are the \emph{atomic programs} $a, b, c, \ldots$ (also called \emph{atomic actions}), as well as the constants $\id$ (\emph{skip}) and $\bot$ (\emph{diverge}). The programs are constructed using the operations $;$ (\emph{sequential composition}), $\kwIf$ (\emph{conditional}), $\kwWhile$ (\emph{iteration}), and $\dem$ (\emph{demonic nondeterministic choice}). We write $f, g, h, \ldots$ to range over arbitrary programs. So, the programs are given by the following grammar:
\begin{align*}
\text{programs $f,g$} ::= {}
&\text{atomic actions $a, b, c, \ldots$} \mid
\id \mid \bot \mid {}
\\
&f;g \mid
\ifThenElse{p}{f}{g} \mid
\whileDo{p}{f} \mid f \dem g.
\end{align*}
For brevity, we also write $p[f,g]$ instead of $\ifThenElse{p}{f}{g}$, and $\wh p f$ instead of $\whileDo{p}{f}$.
\end{defi}

In order to give meaning to these abstract while programs, we first need to specify a nonempty set $S$ representing the state space. Additionally, we need to know how the atomic actions $a, b, c, \ldots$ transform the program state, and which states satisfy an atomic test $p$. So, for every atomic test we are given a subset $R(p) \subseteq S$ of the states that satisfy $p$. Moreover, for every action $a$ assume that we are given a function $R(a): S \to \wp S$, where $\wp S$ is the powerset of $S$. If $u$ and $v$ are states in $S$ with $v \in R(a)(u)$, then we understand this as saying that: executing the action $a$ when in state $u$ may result in a final state $v$. It remains now to describe how an arbitrary program scheme computes. The intended semantics is \emph{operational} and it gives us all the intermediate steps of the computation. A configuration is a pair $(u,f)$ of a state $u$ and a program $f$ and $\to$ is a relation on configurations that describes one step of the computation. A configuration $(u,\id)$ is \emph{final}, which means that the computation halts. We see in Figure~\ref{fig:operationalWPS} the standard definition of the computation relation, where we have assumed w.l.o.g.\ that $;$ is associative.

\begin{figure}[t]
\centering
\addtolength{\jot}{-0.3ex}
$\begin{gathered}
\begin{aligned}[t]
(u,a) &\to (v,\id),\, 
\text{for $v \in R(a)(u)$}
\\
(u,\id) &\to
\\
(u,\bot) &\to (u,\bot)
\\
(u,p[f,g]) &\to (u,f),\, \text{if $u \in R(p)$}
\\
(u,p[f,g]) &\to (u,g),\, \text{if $u \notin R(p)$}
\\
(u,\wh p f) &\to (u,f;\wh p f),\, \text{if $u \in R(p)$}
\\
(u,\wh p f) &\to (u,\id),\, \text{if $u \notin R(p)$}
\\
(u,f \dem g) &\to (u,f),\ (u,g)
\end{aligned}
\qquad
\begin{aligned}[t]
(u,a;h) &\to (v,\id;h),\, 
\text{for $v \in R(a)(u)$}
\\
(u,\id;h) &\to (u,h)
\\
(u,\bot;h) &\to (u,\bot;h)
\\
(u,p[f,g];h) &\to (a,f;h),\, \text{if $u \in R(p)$}
\\
(u,p[f,g];h) &\to (a,g;h),\, \text{if $u \notin R(p)$}
\\
(u,(\wh p f);h) &\to (u,f;(\wh p f);h),\, \text{if $u \notin R(p)$}
\\
(u,(\wh p f);h) &\to (u,\id;h),\, \text{if $u \notin R(p)$}
\\
(u,(f \dem g);h) &\to
(u,f;h),\ (\alpha,g;h)
\end{aligned}
\end{gathered}$
\caption{While Program Schemes: The standard operational model for the interpretation $R$ of atomic symbols.}
\label{fig:operationalWPS}
\end{figure}

The operational semantics of Figure~\ref{fig:operationalWPS} describes fully how a program executes, but for our later logical investigation this description carries too much irrelevant information. We would instead like to focus on the \emph{input-output} behavior of a program $f$. We thus \emph{summarize} the meaning of $f$ as a function $R(f): S \to \wp S$, which is defined as follows:
\[
  v \in R(f)(u) \stackrel{\text{def}}{\iff}
  (u,f) \to \cdots \to (v,\id).
\]
The right-hand side of the above equivalence says that there is a sequence of computation steps from the initial configuration $(u,f)$ to the final configuration $(\id,v)$. These input-output summaries $R(f): S \to \wp S$ constitute the standard \emph{denotational semantics} of nondeterministic while program schemes, also known as the \emph{relational semantics} of programs. It is a very pleasant fact that the functions $R(f)$ have a straightforward \emph{compositional} definition, namely by induction on the structure of $f$. This result is completely standard, and it asserts that denotational equality coincides with operational equivalence. This property is sometimes dubbed as \emph{full abstraction}.

Before we give the formal denotational semantics of while program schemes, we need to define some useful notation. In particular, we will consider an algebra of binary relations (equivalently, their representation as ``nondeterministic functions'') with operations that can give direct meaning to the syntactic constructors of program schemes.

\begin{defi}[Nondeterministic Functions \& Operations]
\label{def:nondet}
For a set $S$, we write $\wp S$ to denote the \emph{powerset} of $S$. A function of type $k: S \to \wp S$ is a \emph{nondeterministic function} on $S$. We also use the notation $k: S \nto S$. We write $k: u \mapsto v$ to mean that $v \in k(u)$. We think informally that such a function describes only one kind of nondeterminism (for our purposes here, demonic nondeterminism). Consider the operations of Figure~\ref{fig:nondetOps}. The choice operation $\plus$ induces a partial order $\leq$ on $S \nto S$ given by : $k \leq \ell$ iff $k \plus \ell = \ell$.
\end{defi}

\begin{figure}
\centering
$\begin{aligned}
&\text{(Kleisli) composition $\kc$}
\qquad&\qquad
(k \kc \ell)(u) &\triangleq
\tbigcup_{v \in k(u)} \ell(v)
\\
&\text{Conditional $(\cdot)\lb -,- \rb$}
&
P \lb k,\ell \rb(u) &\triangleq
k(u),\ \text{if $u \in P$}
\\
&&
P \lb k,\ell \rb(u) &\triangleq
\ell(u),\ \text{if $u \notin P$}
\\
&\text{Binary choice $\plus$}
&
(k \plus \ell)(u) &\triangleq
k(u) \cup \ell(u)
\\
&\text{Arbitrary choice $\tsum$}
&
\bigl( \tbigplus_{i \in J} k_i \bigr)(u) &\triangleq
\tbigcup_{i \in J} k_i(u)
\\
&\text{Identity $1_S$}
&
1_S(u) &\triangleq \{ u \}
\\
&\text{Zero $0_S$}
&
0_S(u) &\triangleq \emptyset
\\
&\text{Iteration $(\sWhileDo \cdot -)$}
&
\sWhileDo P k &\triangleq
\tbigplus_{n \geq 0} V_n,\ \text{where}
\\
&&
V_0 &\triangleq
P \lb 0_S,1_S \rb
\\
&&
V_{n+1} &\triangleq P \lb k \kc V_n,1_S \rb
\end{aligned}$
\caption{Semantic operations for nondeterministic functions $S \nto S$.}
\label{fig:nondetOps}
\end{figure}


\begin{defi}[Nondeterministic Interpretation of Program Schemes]
\label{def:nondetI}
An interpretation of the language of nondeterministic while program schemes consists of a nonempty set $S$, called the \emph{state space}, and an \emph{interpretation function} $R$. The elements of $S$ are called \emph{states}, and we will be using letters $u, v, w, \ldots$ to range over them. For a program term $f$, its \emph{interpretation} $R(f): S \nto S$ is a nondeterministic function on $S$.

The interpretation $R(p)$ of a test $p$ is a unary predicate on $S$, i.e., $R(p) \subseteq S$.
$R$ specifies the meaning of every atomic test, and it extends as follows:
\begin{align*}
R(\true) &= S
&
R(\neg p) &= \compl R(p)
&
R(p \land q) &= R(p) \cap R(q)
\\
R(\false) &= \emptyset
&&&
R(p \lor q) &= R(p) \cup R(q)
\end{align*}
where $\compl$ is the operation of complementation w.r.t.\ $S$, that is, $\compl A = S \setminus A$. Moreover, the interpretation function $R$ specifies the meaning $R(a): S \nto S$ of every atomic program. We extend the interpretation to all program terms:
\begin{align*}
R(\id) &= 1_S
&
R(f;g) &= R(f) \kc R(g)
&
R(p[f,g]) &=
R(p) \lb R(f),R(g) \rb
\\
R(\bot) &= 0_S
&
R(f \dem g) &= R(f) \plus R(g)
&
R(\wh p f) &= \sWhileDo{R(p)}{R(f)}
\end{align*}
Our definition agrees with the standard relational semantics of while schemes.
\end{defi}

\section{The Operational Semantics of Dual Nondeterminism}
\label{sec:operational}

We extend the syntax of nondeterministic program schemes with the additional construct $\ang$ of \emph{angelic (nondeterministic) choice}. So, the grammar for the program terms now becomes:
\begin{align*}
\text{programs $f$, $g$} ::=
\text{actions $a, b, \ldots$} \mid
\id \mid \bot \mid
f;g \mid
p[f,g] \mid
\wh p f \mid
f \dem g \mid f \ang g.
\end{align*}
We call these program terms \emph{while game schemes}, because they can be considered to be descriptions of games between the angel (who controls the angelic choices) and the demon (who controls the demonic choices). Informally, the angel tries to satisfy the specification, while the demon attempts to falsify it.

We consider two-player games between the \emph{existential} player $\exists$ (angel) and the \emph{universal} player $\forall$ (demon). The games are played on arenas of arbitrary cardinality and are of infinite duration. If $\sigma$ is a player, then $\neg\sigma$ is the other player. Such games are considered extensively in the literature for the verification of reactive systems, see for example \cite{thomas1995}. The following definition of safety games (Definition~\ref{def:safetyGames}) slightly modifies the definition of \cite{thomas1995} in order to fit our setting more naturally.

\begin{defi}[Safety Games]
\label{def:safetyGames}
A \emph{safety game} is a tuple $G = (V,V_\exists,V_\forall,\to,E)$, where $V$ is the set of all vertices, $V_\exists$ is the set of $\exists$-vertices (which belong to the existential player), $V_\forall$ is the set of $\forall$-vertices (which belong to the universal player), $V_\exists$ and $V_\forall$ are disjoint subsets of $V$, $\to$ is a binary \emph{transition relation} on $V$, and $E \subseteq V$ is the set of \emph{error vertices}. We use the letters $u,v,w,\ldots$ to range over vertices in $V$, and we write $u \to v$ to mean that the pair $(u,v)$ belongs to the transition relation. We require additionally that every vertex has a successor, and that the vertices $V_? = V \setminus (V_\exists \cup V_\forall)$ that belong to no player have exactly one successor. The last requirement says equivalently that if a vertex has more than one successor, then it must belong to one of the players.

We need to introduce some terminology, which is to be understood with respect to a specific game. A \emph{position} is a finite nonempty path, and a \emph{play} is an infinite path. A $u$-position ($u$-play) is a position (play) that starts from vertex $u$. We say that Player $\exists$ \emph{wins} a play if no error vertex appears in it. Player $\forall$ wins if the play contains an error vertex. A \emph{strategy for Player $\sigma$} or a \emph{$\sigma$-strategy} is a function that maps every position ending in a $\sigma$-vertex $u$ to one of the successors of $u$. In a \emph{memoryless} or \emph{positional} strategy the choice depends only on the last vertex. So, we can represent a memoryless strategy for Player $\sigma$ as a function that maps every $\sigma$-vertex to one of its successors. We say that a path \emph{conforms} to a $\sigma$-strategy $f_\sigma$ if every transition from a $\sigma$-vertex in the path is the one prescribed by the strategy $f_\sigma$. A $(u,f_\sigma)$-position is a $u$-position that conforms to the strategy $f_\sigma$. We define a $(u,f_\sigma)$-play similarly. A $(u,f_\exists,f_\forall)$-position is a $u$-position that conforms to both $f_\exists$ and $f_\forall$. A $(u,f_\exists,f_\forall)$-play is defined similarly. We denote by $\play(u,f_\exists,f_\forall)$ the unique $(u,f_\exists,f_\forall)$-play, which is the infinite path formed by starting at vertex $u$ and then following the strategies $f_\exists$ and $f_\forall$ for every transition allowing more than one choice.

We say that a set of vertices $U \subseteq V$ is
\emph{$\sigma$-closed} if
\begin{enumerate}[label=(\roman*)]
\item
every vertex of $V_? \cap U$ has its unique successor in $U$,
\item
every $\sigma$-vertex of $U$ has at least one successor in $U$, and
\item
every $\neg\sigma$-vertex of $U$ has all of its successors in $U$.
\end{enumerate}
\end{defi}

\begin{defi}[Winning Regions]
\label{def:winning}
Given a safety game $G = (V,V_\exists,V_\forall,\to,E)$, we will define the sets $W_\exists \subseteq V$ and $W_\forall \subseteq V$, which partition the set $V$ of vertices. The set $W_\exists$ is called the \emph{winning region} of Player $\exists$, and $W_\forall$ is the \emph{winning region} of Player $\forall$. First, we define the transfinite sequence $(W_\forall^\kappa)_{\kappa \in \Ord}$ of sets. We write $\Ord$ for the class of ordinals. Informally, for an ordinal $\kappa$, the set $W_\forall^\kappa$ consists of the nodes from which Player $\forall$ can force a visit to $E$ in at most $\kappa$ steps.
\begin{align*}
W_\forall^0 &\triangleq E
&
W_\forall^{\kappa+1} &\triangleq
W_\forall^\kappa \cup
\begin{aligned}[t]
&\{ u \in V_? \mid
    \text{the unique successor of $u$ is in $W_\forall^\kappa$}
 \} \cup {}
\\
&\{ u \in V_\exists \mid
   \text{every successor of $u$ is in $W_\forall^\kappa$}
\}
\cup {}
\\
&\{ u \in V_\forall \mid
   \text{some successor of $u$ is in $W_\forall^\kappa$}
\}
\end{aligned}
\\
&&
W_\forall^\lambda &\triangleq
\tbigcup_{\kappa < \lambda} W_\forall^\kappa,
\ \text{for a limit ordinal $\lambda$}
\end{align*}
Now, we can define the winning regions of the players in terms of the above sequence:
\begin{align*}
W_\forall &\triangleq
\tbigcup_{\kappa \in \Ord} W_\forall^\kappa
&
W_\exists &\triangleq
V \setminus W_\forall
\end{align*}
Notice that the sets $W_\forall^0 \subseteq W_\forall^1 \subseteq \cdots \subseteq W_\forall^\kappa \subseteq \cdots$ form a transfinite chain w.r.t.\ inclusion.
\end{defi}

\begin{theorem}[Memoryless Determinacy]
\label{thm:determined}
Let $G = (V,V_\exists,V_\forall,\to,E)$ be a safety game, and $W_\exists$, $W_\forall$ be the winning regions of the two players. There is a memoryless $\exists$-strategy $f_\exists\star$ and a memoryless $\forall$-strategy $f_\forall\star$ that witness uniformly the winning regions. That is:
\begin{enumerate}
\item
For every $u \in W_\exists$ and every $\forall$-strategy $f_\forall$, $\play(u,f_\exists\star,f_\forall)$ is won by Player $\exists$.
\item
For every $u \in W_\forall$ and every $\exists$-strategy $f_\exists$, $\play(u,f_\exists,f_\forall\star)$ is won by Player $\forall$.
\end{enumerate}
\end{theorem}
\begin{proof}[Proof sketch]
The idea for Part (1) is to show that the set $W_\exists$ is $\exists$-closed, and therefore Player $\exists$ has a memoryless strategy $f\star_\exists$ that keeps within $W_\exists$ every play starting from a vertex of $W_\exists$. For the sake of contradiction, assume that $u \in W_\exists$ is a vertex which witnesses that $W_\exists$ is \emph{not} $\exists$-closed. There are three distinct possibilities for $u$:
\begin{enumerate}[label=(\roman*)]
\item
$u \in V_?$ and its unique successor is in $W_\forall$, or
\item
$u \in V_\exists$ and every successor of $u$ is in $W_\forall$, or
\item
$u \in V_\forall$ and some successor of $u$ is in $W_\forall$.
\end{enumerate}
Every possibility implies that $u \in W_\forall$, which gives the desired contradiction. So, $W_\exists$ is indeed $\exists$-closed. For Part (2), the proof is based on labeling every vertex $u \in W_\forall$ as follows:
\[
  \ord(u) \triangleq
  \text{the least ordinal $\kappa$ such that $u \in W_\forall^\kappa$}.
\]
One can then show that Player $\forall$ has a strategy $f\star_\forall$ so that for every play that starts from a vertex of $W_\forall$ the labels keep going down until eventually an error vertex is reached.
\end{proof}

\begin{observation}[Summarizing Safety Games]
\label{obs:summary}
We have already discussed in \S\ref{sec:prelim} that a denotational semantics is most useful when it is a \emph{faithful summarization} of the intended operational meaning. Before presenting a denotational semantics of dual nondeterminism in \S\ref{sec:denotational} we will discuss here what constitutes a summarization for safety games, and what kind of mathematical objects are useful for this purpose.


Consider a safety game $(V,V_\exists,V_\forall,\to,E)$ and recall that $W_\exists$ is the set of vertices from which the existential player (angel) has a strategy to avoid the error vertices. We write $W_\exists(E)$ to emphasize the fact that the winning region of Player $\exists$ depends on which vertices are designated as error vertices. Theorem~\ref{thm:determined} implies that:
\begin{center}
\em
If $u \in W_\exists(E)$ then the angel can keep any $u$-play within the non-error vertices $\compl E$.
\end{center}
Let us think about the more general situation, where the error vertices $E$ can be varied. We can summarize the guarantees that the angel can make with the following object:
\[
  \phi \triangleq
  \{ (u,\compl E) \mid
     \text{in the game $(V,V_\exists,V_\forall,\to,E)$, the vertex $u$ is in $W_\exists(E)$}
  \}.
\]
Immediately from the definition of the winning regions (see Definition~\ref{def:winning}) we see that:
\begin{enumerate}
\item
The inclusion $E_1 \subseteq E_2$ implies $W_\forall(E_1) \subseteq W_\forall(E_2)$ and therefore $W_\exists(E_2) \subseteq W_\exists(E_1)$. Assuming that $X \subseteq Y \subseteq V$ we have that $\compl Y \subseteq \compl X$ and
\[
  (u,X) \in \phi \implies
  u \in W_\exists(\compl X) \implies
  u \in W_\exists(\compl Y) \implies
  (u,Y) \in \phi.
\]
\item
Notice that for error vertices $E = \emptyset$ we have that $W_\forall(\emptyset) = \emptyset$ and hence $W_\exists(\emptyset) = V$. It follows that $(u,V)$ belongs to $\phi$.
\end{enumerate}
Both of the above properties will turn out to be crucial for our development, and they motivate the notion of a \emph{game function} given formally in Definition~\ref{def:gameFunction} of \S\ref{sec:denotational}. For the rest of this section, it suffices to keep in mind that the denotations of game schemes will be binary relations from $S$ to $\wp S$, where $S$ is the state space.
\end{observation}

In order to streamline the presentation of the operational semantics, we should make a couple of inconsequential modifications to the language of game schemes. We restrict slightly the syntax of program terms by eliminating the diverging $\bot$ program, and by forbidding compositions $(f;g);h$ that associate to the left. These are not really limitations, because for every reasonable semantics $\bot$ has to be equivalent to the infinite loop $\whileDo \true \id$, and $(f;g);h$ has to be equivalent to $f;(g;h)$. So, we define the syntactic categories \emph{factor} and \emph{term} with the following grammars:
\begin{align*}
\text{factor $e$} &::=
\text{atomic program $a$, $b$, \ldots} \mid \id \mid
p[f,g] \mid \wh p f \mid
f \ang g \mid f \dem g
\\
\text{terms $f, g$} &::=
e \mid e;f
\end{align*}
According to the above definition, a term  is a nonempty list of factors. We write $@$ for the operation of list concatenation: $e@g = e;g$ and $(e;f)@g = e;(f@g)$.

\begin{defi}[Closure \& The $\tto$ Relation On Terms]
We define the \emph{closure} map $\cl(\cdot)$, which sends a program term to a finite set of program terms.
\begin{align*}
\cl(a) &= \{ a,\id \}
&
\cl(\wh p f) &= \{ \wh p f, \id \} \cup \cl(f)@\wh p f
&
\cl(e;f) &= \cl(e)@f \cup \cl(f)
\\
\cl(\id) &= \{ \id \}
&\hspace{-0.25em}
\cl(f \oplus g) &=
\{ f \oplus g \} \cup \cl(f) \cup \cl(g)
\end{align*}
where $\oplus$ is any of the constructors $\ang$, $\dem$, or $p[-,-]$. If $F$ is a set of terms and $g$ is a term, we lift the concatenation operation $@$ as follows: $F@g = \{ f@g \mid f \in F \}$. Now, we define the \emph{one-step reachability relation} $\tto$ on program terms as follows:
\begin{align*}
a &\tto \id
&
a;h &\tto \id;h
\\
\id &\tto
&
\id;h &\tto h
\\
f \oplus g &\tto f,\, g
&
(f \oplus g);h &\tto f@h,\, g@h
\\
\wh p f &\tto f@\wh p f,\, \id
&
\wh p f; h &\tto f@(\wh p f);h,\, \id;h
\end{align*}
The above definition of $\tto$ says, in particular, that $\id$ has no successor. The while loop $\wh p f$ has exactly two successors, namely $f@\wh p f$ and $\id$. We write $\tto\star$ to denote the reflexive transitive closure of the relation $\tto$.
\end{defi}

\begin{lemma}
\label{lemma:closureReach}
The following hold for the closure map and the reachability relation:
\begin{enumerate}
\item
\label{part:size}
Let $f$ be a program term. The cardinality of the set $C(f)$ is linear in the size $|f|$ of the term $f$. More specifically, it holds that $|C(f)| \leq 2 |f|$.
\item
\label{part:oneStep}
For terms $f, f'$ and $g$, if $f \tto f'$ then $f@g \tto f'@g$.
\item
\label{part:manySteps}
For terms $f, f'$ and $g$, if $f \tto\star f'$ then $f@g \tto\star f'@g$.
\item
\label{part:succ}
For terms $f$ and $g$, the $\tto$-successors of $f@g$ are contained in $\{ g \} \cup \{ f'@g \mid f \tto f' \}$.
\item
\label{part:reachInCl}
For every term $f$, the set $\cl(f)$ contains $f$ and is closed under $\tto$.
\item
\label{part:clInReach}
For all terms $f$ and $f'$, if $f' \in \cl(f)$ then $f \tto\star f'$.
\item
\label{part:clEqReach}
Let $f$ be a program term. Then, $C(f)$ is equal to the set $\{ f' \mid f \tto\star f' \}$ of terms that are reachable from $f$ via $\tto$.
\end{enumerate}
\textbf{\em Note}: Parts \eqref{part:size} and \eqref{part:clEqReach} are the main properties that we will need later. Parts \eqref{part:oneStep}--\eqref{part:clInReach} are the intermediate claims that are needed to obtain Part \eqref{part:clEqReach}.
\end{lemma}
\begin{proof}
Part \eqref{part:size} can be shown by induction on the structure of $f$. Parts \eqref{part:oneStep} and \eqref{part:succ} are proved with a case analysis on the form of the term $f$. Part \eqref{part:manySteps} follows from Part \eqref{part:oneStep} by induction on the length of the $\tto$-sequence. Part \eqref{part:reachInCl} is shown by induction on $f$, making use of Part \eqref{part:succ}. The proof of Part \eqref{part:clInReach} requires an induction on $f$ and Part \eqref{part:manySteps}. Part \eqref{part:clEqReach} is an immediate consequence of Part \eqref{part:reachInCl} and Part \eqref{part:clInReach}.
\end{proof}

\begin{figure}[t]
\centering
\addtolength{\jot}{-0.3ex}
$\begin{gathered}
\begin{aligned}[t]
(u,a) &\to (X,\id),\, 
\text{when $(u,X) \in I(a)$}
\\
(u,\id) &\to 
\\
(u,p[f,g]) &\to (u,f),\, \text{if $u \in I(p)$}
\\
(u,p[f,g]) &\to (u,g),\, \text{if $u \notin I(p)$}
\\
(u,\wh p f) &\to (u,f@\wh p f),\, \text{if $u \in I(p)$}
\\
(u,\wh p f) &\to (u,\id),\, \text{if $u \notin I(p)$}
\\
(u,f \ang g) &\to (u,f),\ (u,g)
\\
(u,f \dem g) &\to (u,f),\ (u,g)
\end{aligned}
\ 
\begin{aligned}[t]
(u,a;h) &\to (X,\id;h),\, 
\text{when $(u,X) \in I(a)$}
\\
(u,\id;h) &\to (u,h)
\\
(u,p[f,g];h) &\to (a,f@h),\, \text{if $u \in I(p)$}
\\
(u,p[f,g];h) &\to (a,g@h),\, \text{if $u \notin I(p)$}
\\
(u,(\wh p f);h) &\to (u,f@(\wh p f);h),\, \text{if $u \notin I(p)$}
\\
(u,(\wh p f);h) &\to (u,\id;h),\, \text{if $u \notin I(p)$}
\\
(u,(f \ang g);h) &\to
(u,f@h),\ (u,g@h)
\\
(u,(f \dem g);h) &\to
(u,f@h),\ (u,g@h)
\end{aligned}
\\
(X,f) \to (v,f),
\ \text{where $v \in X \subseteq S$}
\end{gathered}$
\caption{While Game Schemes: Operational model for interpretation $I$ of atomic symbols.}
\label{fig:operationalWGS}
\end{figure}

\begin{defi}[Operational Model for Game Schemes]
\label{def:opModel}
Let $S$ be a nonempty set of states, and $I$ be an interpretation function for the atomic tests and actions. That is, $I$ specifies a unary predicate $I(p) \subseteq S$ for every atomic test $p$, and a binary relation $I(a) \subseteq S \times \wp S$ for every atomic action $a$. Let $f$ be a program term, and $E \subseteq S$ be a set of \emph{error states}. We define the \emph{operational model} for $I, f, E$, denoted $G_I(f,E)$, to be the safety game
\begin{align*}
G_I(f,E) &=
(V,V_\exists,V_\forall,\to,E \times \{ \id \}), \ \text{where}
\\
V &=
(S \times C(f)) \cup
(\mathcal{X} \times C(f))\ \text{with}
\\
\mathcal{X} &= \{ X \subseteq S \mid
  \text{$(u,X) \in I(a)$ for some $a \in C(f)$ and $u \in S$}
\},
\end{align*}
and the transition relation $\to$ is defined in Figure~\ref{fig:operationalWGS}. Part~\eqref{part:clEqReach} of Lemma~\ref{lemma:closureReach} implies that $V$ is closed under $\to$ (note that $\tto$ is the ``projection'' of $\to$ to the second component). Strictly speaking, in order for $G_I(f,E)$ to be a safety game according to Definition~\ref{def:safetyGames}, we would need to modify $\to$ so that every vertex $(u,\id)$ has a self-loop instead of being a sink, but this would be an inconsequential modification. For the components $V_\exists$ and $V_\forall$ we put:
\begin{itemize}[label=$-$]
\item
The $\exists$-vertices $V_\exists \subseteq V$ consist of the pairs of the form $(u,f \ang g)$, as well as the pairs $(u,a)$ and $(u,a;h)$ for atomic program $a$.
\item
The $\forall$-vertices $V_\forall \subseteq V$ consist of the pairs $(u,f \dem g)$, as well as the pairs $(X,f)$ where $(u,X) \in I(a)$ for some atomic action $a$ and state $u$.
\end{itemize}
We think of the pairs $(u,\id)$ as being \emph{terminal vertices}, and the error vertices are $E \times \{ \id \}$.
\end{defi}

\newcommand{\inc}{x\texttt{++}}
\begin{figure}[t]
\centering
$\begin{gathered}
\begin{tikzpicture}[baseline=(current bounding box.center), ->, node distance=20mm, auto]
  \small
  \node (N11) {$(0,h)$};
  \node (N12) [rounded rectangle, draw, right of=N11, node distance=17mm] {$0,f;g;h$};
  \node (N13) [right of=N12] {$(0,\id;g;h)$};
  \node (N23) [below of=N13, node distance=7mm] {$(0,\inc;g;h)$};
  \node (N14) [rectangle, draw, right of=N13] {$0,g;h$};
  \node (N24) [rectangle, draw, below of=N14, node distance=8mm] {$1,g;h$};
  \node (N15) [right of=N14] {$(0,\inc;h)$};
  \node (N05) [above of=N15, node distance=6mm] {$(0,\id;h)$};
  \node (N25) [right of=N24] {$(1,\id;h)$};
  \node (N35) [below of=N25, node distance=6mm] {$(1,\inc;h)$};
  \node (N16) [right of=N15, node distance=18mm] {$(1,h)$};
  \node (N36) [right of=N35, node distance=18mm] {$(2,h)$};
  \node (N17) [right of=N16, node distance=14mm] {$(1,\id)$};
  \node (N37) [right of=N36, node distance=14mm] {$(2,\id)$};

  \path (N11) edge (N12);
  \path (N12) edge (N13);
  \path (N12) edge (N23);
  \path (N13) edge (N14);
  \path (N23) edge (N24);
  \path (N14) edge (N05);
  \path (N14) edge (N15);
  \path (N24) edge (N25);
  \path (N24) edge (N35);
  \path (N05) edge[bend right=11] (N11);
  \path (N15) edge (N16);
  \path (N25) edge (N16);
  \path (N35) edge (N36);
  \path (N16) edge (N17);
  \path (N36) edge (N37);
%
\end{tikzpicture}
\hspace{2em}
\begin{aligned}
f &= \id \ang \inc
\\
g &= \id \dem \inc
\\
p &= (x=0)
\\
h &= \wh p (f;g)
\end{aligned}
\end{gathered}$
\caption{Reduced operational model for the dually nondeterministic program $h$. The vertices of the demon (angel) are indicated with rectangles (rounded rectangles).}
\label{fig:exOperational}
\end{figure}

\begin{example}
\label{ex:operational}
Suppose that we want to describe a program whose state consists of a single variable $x$ that can take values 0, 1 or 2. The only atomic action that we consider is $\inc$, which assigns $(x+1) \bmod 3$ to the variable $x$. The atomic test $(x=0)$ checks if the value of $x$ is equal to 0. Consider the program
\[
  h \triangleq \whileDo{(x=0)}{((\id \ang \inc);(\id \dem \inc))}.
\]
On the right-hand side of Figure~\ref{fig:exOperational} we have some abbreviations for parts of the program, and on the left-hand side we see a simplified version of the operational model. We have only drawn the vertices that are reachable from $(0,h)$, $(1,h)$ and $(2,h)$. Since the action $\inc$ is deterministic, we have also made some simplifications such as: the transition sequence $(0,\inc;h) \to (\{1\},\id;h) \to (1,h)$ has been reduced to $(0,\inc;h) \to (1,h)$.

The terminal vertices shown in Figure~\ref{fig:exOperational} are $(1,\id)$ and $(2,\id)$. Suppose that $(2,\id)$ is the unique error vertex. The winning region $W_\forall$ of the demon consists of:
\begin{align*}
(2,\id) &&
(2,h) &&
(1,\inc;h) &&
(1,g;h) &&
(0,\inc;g;h)
\end{align*}
The rest of the vertices form the winning region $W_\exists$ of the angel.
\end{example}

\section{Denotational Semantics and Full Abstraction}
\label{sec:denotational}

In \S\ref{sec:operational} we presented the syntax of while game schemes and we gave an operational model based on two-player games on finite graphs. Because of this adversarial dynamics, the input-output behavior can no longer be described using binary relations consisting of the possible input-ouput pairs, as is done for usual programs (recall Definition~\ref{def:nondetI}). Instead, we will adopt an angel-centric view, and we will record in our program denotations the predicates that the angel can guarantee of the output. As usual, a nonempty set $S$ represents the abstract state space, and every test is interpreted as a unary predicate on the state space. Every program term is now interpreted as a binary relation from $S$ to $\wp S$.

Consider such a binary relation $\phi \subseteq S \times \wp S$, which should be thought of as the extension of a dually nondeterministic program. Informally, the pair $(u,X)$ is supposed to belong to $\phi$ when the following holds: if the program starts at state $u$, then the angel has a strategy so that whatever the demon does, the final state (supposing that the program terminates) satisfies the predicate $X$.

The binary relation $\phi \subseteq S \times \wp S$ encodes both the choices of the angel and the demon, and it can be understood intuitively as a two-round game. The angel moves first, and then the demon makes the final move. The options that are available to the angel are given by multiple pairs $(u,X_1)$, $(u,X_2)$, and so on. So, when the game starts at state $u$, the angel first chooses either $X_1$, or $X_2$, or any of the other available options. Suppose that the angel first chooses $X_i$, where $(u,X_i)$ is in $\phi$. Then, during the second round, the demon chooses some final state $v \in X_i$. See Figure~\ref{fig:game} for a visualization of this game.

\begin{figure}[t]
\centering
$\begin{gathered}
\begin{tikzpicture}[baseline=(current bounding box.center), ->, node distance=8mm, auto]
  \small
  \node (S) {};
  \node (A) [circle, draw, right of=S, node distance=15mm] {$u$};
  \node (D) [draw, right of=A, node distance=17mm] {$X_2$};
  \node (D1) [draw, above of=D] {$X_1$};
  \node (D2) [draw, below of=D] {$X_3$};
  \node (DD1) [right of=D1, node distance=17mm] {};
  \node (E1) [above of=DD1, node distance=2mm] {$v_1$};
  \node (E2) [below of=DD1, node distance=2mm] {$v_2$};
  \node (DD) [right of=D, node distance=17mm] {$v_3$};
  \node (DD2) [right of=D2, node distance=17mm] {};
  \node (E4) [above of=DD2, node distance=2mm] {$v_4$};
  \node (E5) [below of=DD2, node distance=2mm] {$v_5$};

  \path (S) edge[thick] node {start} (A);
  \path (A) edge (D1);
  \path (A) edge (D);
  \path (A) edge (D2);
  \path (D) edge (DD);
  \path (D1) edge (E1);
  \path (D1) edge (E2);
  \path (D2) edge (E4);
  \path (D2) edge (E5);
\end{tikzpicture}
\hspace{4em}
\begin{aligned}
\vspace{-1ex}
\phi &= \{ (u,X_1), (u,X_2), (u,X_3) \}
\\
X_1 &= \{ v_1, v_2 \}
\\
X_2 &= \{ v_3 \}
\\
X_3 &= \{ v_4, v_5 \}
\end{aligned}
\end{gathered}$
\caption{Visualization of a two-round game between the angel and the demon, as described by a relation $\phi \subseteq S \times \wp S$. The angel moves at the circled node, the demon moves at the boxed nodes, and the nodes with no outline are terminal.}
\label{fig:game}
\end{figure}

When $(u,X)$ is in $\phi$, we understand this as meaning that that the angel can guarantee the predicate $X$ when we start at $u$. So, it is reasonable to expect that the angel also guarantees from $u$ any predicate that is weaker than $X$. In order to be consistent with the viewpoint of partial correctness, we also want to require that the angel can guarantee anything in the case of nontermination. Recall Observation~\ref{obs:summary}, where we discuss how to summarize two-player games on graphs from the perspective of what the angel can guarantee. These considerations motivate the following definition.

\begin{defi}[Game Functions]
\label{def:gameFunction}
Let $S$ be a nonempty set called the \emph{state space}. We say that $\phi \subseteq S \times \wp S$ is a \emph{game function} on $S$, denoted $\phi: S \gto S$, if it satifies:
\begin{enumerate}
\item
The set $\phi$ is \emph{closed upwards}, which is defined to mean the following:
\[
  \text{$(u,X) \in \phi$ and $X \subseteq Y$}
  \implies
  (u,Y) \in \phi
\]
for every state $u \in S$ and all predicates $X, Y \subseteq S$.
\item
\emph{Non-emptiness}:
For every $u \in S$ there is some $X \subseteq S$ with $(u,X) \in \phi$.
\end{enumerate}
Given Condition (1), we can equivalently require that $(u,S) \in \phi$ for every $u \in S$, instead of having Condition (2). This essentially says that the angel always guarantees that the output lies in the state space.

Let $\phi: S \gto S$ be a game function. The \emph{options} of the angel at a state $u \in S$, which we denote by $\phi(u)$, is the collection of predicates
\[
  \phi(u) = \{ X \subseteq S \mid (u,X) \in \phi \}.
\]
In other words, $\phi(u)$ is the set of all predicates that the angel can guarantee from $u$. This notation suggests that we can equivalently understand $\phi$ as being a function $S \to \wp\wp S$. Indeed, the definition says that $(u,X) \in \phi$ iff $X \in \phi(u)$ for all $u \in S$ and $X \subseteq S$.

\end{defi}

Now, we will observe that the space of game functions is large enough to encompass nondeterministic functions as a special case. To make this claim precise, we need to define a \emph{lifting} operation, which embeds the nondeterministic functions into the game functions. As we will see, this is not merely an injective map, but it also commutes with the corresponding semantic operations in these two spaces. So, the algebra of nondeterministic functions is embedded via the lifting map into the algebra of game functions.

\begin{defi}[Lifting \& Non-Angelic Game Functions]
\label{def:lift}
Let $S$ be a state space, and $k: S \nto S$ be a nondeterministic function on $S$. We define the \emph{lifting} of $k$ to be the game function $\lift k: S \gto S$, which is given by
\[
  \lift k \triangleq
  \{ (u,Y) \mid
     \text{$u \in A$ and $k(u) \subseteq Y$}
  \}: S \gto S.
\]
This says that for every state $u \in S$ and predicate $Y \subseteq S$: $(u,Y) \in \lift k$ iff $k(u) \subseteq Y$. The lifting operation is thus a mapping from the space $S \nto S$ to $S \gto S$.

We say that a game function $\phi: S \gto S$ is \emph{non-angelic} if it is the lifting of a nondeterministic function, that is, $\phi = \lift k$ for some $k: S \nto S$. Essentially, the definition says that the angel always has exactly one minimal choice: for every $u \in S$ there is exactly one minimal predicate $k(u)$ that the angel can guarantee.
\end{defi}

\begin{observation}[Demonic \& Angelic Lifting]
In Definition~\ref{def:lift} we consider a lifting operation from the space $S \nto S$ to the space $S \gto S$ which interprets the nondeterminism demonically. This works, because a nondeterministic function $k: S \nto S$ records reachability information, i.e.\ what the demon can achieve. So, we could call $\lift$ more descriptively the \emph{demonic lifting} operation. The question then arises of whether we can define an analogous \emph{angelic lifting} operation which interprets the nondeterminism angelically. First, we notice that the space of nondeterministic functions $S \nto S$ with the operations of Figure~\ref{fig:nondetOps} is inappropriate for modeling pure angelic nondeterminism. Since the angel's goal is safety and the angel wins in the case of nontermination of the program, the semantics should record explicitly when the angel can force divergence. The standard relational semantics of \S\ref{sec:prelim}, however, is ``divergence-oblivious'' in the sense of suppressing the information regarding the possibility of divergence. For example, we have that
\[
  0_S + k = k
  \ \text{for every $k: S \nto S$}.
\]
So, in order to define a reasonable angelic lifting one would have to modify the relational semantics of \S\ref{sec:prelim} to record the possibility of nontermination. While this investigation would be interesting mathematically, it is beyond the scope of the present paper. From a practical standpoint, distinguishing the non-angelic game functions (see Definition~\ref{def:lift}) is crucial for the synthesis applications that we consider here. We have to restrict attention to programs where the atomic actions do not involve any angelic choices in order to formulate a reasonable synthesis problem for angelic strategies. Since we are not concerned with the implementation of demonic strategies (the choices of the demon are beyond our control!), the definition of a reasonable angelic lifting operation is of little use here.
\end{observation}

\begin{figure}[t!]
\centering
$\begin{aligned}
&\text{Composition $\gc$}
\quad&\quad
(u,Z) \in (\phi \gc \psi) &\stackrel{\text{def}}{\iff}
\begin{aligned}[t]
&\text{there is $Y \subseteq S$ s.t.\ $(u,Y) \in \phi$,}
\\[-0.5ex]
&\text{and $(v,Z) \in \psi$ for every $v \in Y$.}
\end{aligned}
\\
&\text{Conditional $(\cdot)\glb -,- \grb$}
&
P \glb \phi,\psi \grb &\triangleq
\bigl(
  \phi \cap (P \times \wp S)
\bigr) \cup
\bigl(
  \psi \cap (\compl P \times \wp S)
\bigr)
\\
&&
P \glb \phi,\psi \grb(u) &=
\phi(u),\ \text{if $u \in P$}
\\
&&
P \glb \phi,\psi \grb(u) &=
\psi(u),\ \text{if $u \notin P$}
\\
&\text{Angelic choice $\gang$}
&
\phi \gang \psi &\triangleq
\phi \cup \psi
\\
&\text{Demonic choice $\gdem$}
&
\phi \gdem \psi &\triangleq
\{ (u,X \cup Y) \mid \text{($u,X) \in \phi$ and $(u,Y) \in \psi$} \}
\\
&&
&= \phi \cap \psi
\\
&\text{Identity $\one_S$}
&
\one_S(u) &\triangleq
\{ (u,X) \mid \text{$u \in S$, $X \subseteq S$ and $u \in X$} \}
\\
&\text{Zero $\zero_S$}
&
\zero_S(u) &\triangleq S \times \wp S
\\
&\text{Iteration $(\gWhileDo \cdot -)$}
&
\gWhileDo P \phi &\triangleq
\tbigcap_{\kappa \in \Ord} W_\kappa,\ \text{where}
\\
&&
W_0 &\triangleq P \glb \zero_S, \one_S \grb
\\
&&
W_{\kappa+1} &\triangleq P \glb \phi \gc W_\kappa, \one_S \grb
\\
&&
W_\lambda &\triangleq \tbigcap_{\kappa < \lambda} W_\kappa,\   \text{for limit ordinal $\lambda$}
\end{aligned}$
\caption{Semantic operations for game functions.}
\label{fig:gameOps}
\end{figure}

We list the formal definitions of the semantic operations on game functions $S \gto S$ in Figure~\ref{fig:gameOps}. As expected, the angelic choice operation $\gang$ increases the options available to the angel. The demonic choice operation $\gdem$ increases the options of the demon. The identity $\one_S$ is the smallest game function that contains $(u,\{u\})$ for every state $u \in S$. Informally, this definition says that on input $u$, the angel guarantees output $u$ in the identity game. The intuition for the definition of the zero function $\zero_S$ is that when the program diverges, the demon cannot lead the game to an error state, therefore the angel can guarantee anything. This describes a notion of partial correctness.

\begin{example}
\label{ex:denotational}
We will calculate now the denotation of the program $h$ from Example~\ref{ex:operational}. We write $S$ for the state space, and $I$ for the interpretation of the atomic symbols. We present below a table with the denotations of all subprograms of $h$.
\begin{align*}
p &= (x=0)
&
f &= \id \ang \inc
&
g &= \id \dem \inc
&
h &= \wh p (f;g)
\end{align*}
 Since the options of the angel are closed upwards, it suffices to record the minimal predicates for every state. Define $P = I(p) = \{0\}$, and we have:
\begin{center}
\begin{tabular}{|c|ccccccccc|}
\hline
state &
$\zero_S$
&
$\one_S$
&
$I(\inc)$
&
$\phi = I(f)$
&
$\psi = I(g)$
&
$\phi \gc \psi = I(f;g)$
&
$W_0$
&
$\phi \gc \psi \gc W_0$
&
$W_1$
\\ \hline
0
&
$\emptyset$
&
$\{0\}$
&
$\{1\}$
&
$\{0\}\ \{1\}$
&
$\{0,1\}$
&
$\{0,1\}\ \{1,2\}$
&
$\emptyset$
&
$\{1\}$\ \textcolor{gray}{$\{1,2\}$}
&
$\{1\}$
\\
1
&
$\emptyset$
&
$\{1\}$
&
$\{2\}$
&
$\{1\}\ \{2\}$
&
$\{1,2\}$
&
$\{1,2\}\ \{2,0\}$
&
$\{1\}$
&
\textcolor{gray}{$\{1,2\}$}\ $\{2\}$
&
$\{1\}$
\\
2
&
$\emptyset$
&
$\{2\}$
&
$\{0\}$
&
$\{2\}\ \{0\}$
&
$\{2,0\}$
&
$\{2,0\}\ \{0,1\}$
&
$\{2\}$
&
$\{2\}\ \{1\}$
&
$\{2\}$
\\ \hline
\end{tabular}
\end{center}
where $W_0 = P \glb \zero_S,\one_S \grb$ and $W_1 = P \glb \phi \gc \psi \gc W_0,\one_S \grb$. We leave as an exercise to the reader to verify that $W_2 = P \glb \phi \gc \psi \gc W_1,\one_S \grb = W_1$. It follows that $I(h) = W_2$.
\end{example}

We note that the definition of Figure~\ref{fig:gameOps} gives the while operation as a greatest fixpoint. This is not surprising, because the semantics we consider is meant to be useful for reasoning about \emph{safety properties}. As we will see, this definition agrees with the standard least fixpoint definition of while loops when there is only one kind of nondeterminism (Lemma~\ref{lemma:lift} below). More importantly, we will prove that our definition is \emph{exactly correct}, becauses it agrees with the intended operational semantics of dual nondeterminism (Theorem~\ref{thm:semantics}).


\begin{lemma}[Lifting Commutes With The Semantic Operations]
\label{lemma:lift}
Let $k$ and $\ell$ be nondeterministic functions on $S$, and $P$ be a unary predicate on $S$. Then, the following hold:
\begin{align*}
\lift 0_S &= \zero_S
&
\lift(k \kc \ell) &=
(\lift k) \gc (\lift\ell)
&
\lift(P \lb k,\ell \rb) &=
P \glb \lift k, \lift\ell \grb
\\
\lift 1_S &= \one_S
&
\lift(k \plus \ell) &=
(\lift k) \gdem (\lift\ell)
&
\lift(\sWhileDo P k) &=
\gWhileDo P (\lift k)
\end{align*}
So, the lifting map commutes with all the semantic operations of nondeterministic functions.
\end{lemma}
\begin{proof}
The cases of $0$, $1$, demonic choice and conditionals are straightforward and we omit them. For the case of composition we have that:
\begin{allowdisplaybreaks}
\begin{align*}
&(u,Z) \in \lift(k \kc \ell) \iff
&&\text{[def.\ of $\lift$]}
\\
&(k \kc \ell)(u) \subseteq Z \iff
&&\text{[def.\ of $\kc$]}
\\
&\tbigcup_{v \in k(u)} \ell(v) \subseteq Z \iff
&&\text{[union and $\subseteq$]}
\\
&\text{$\ell(v) \subseteq Z$ for every $v \in k(u)$} \iff
&&\text{[for ``$\Rightarrow$'' put $Y = k(u)$]}
\\
&\text{$\exists Y \subseteq S.\ $ $k(u) \subseteq Y$ and $\ell(v) \subseteq Z$ for all $v \in Y$} \iff
&&\text{[def.\ of $\lift$]}
\\
&\text{$\exists Y \subseteq S.\ $ $(u,Y) \in \lift k$ and $(v,Z) \in \lift\ell$ for all $v \in Y$} \iff
&&\text{[def.\ of $\gc$]}
\\
&(u,Z) \in (\lift k) \gc (\lift \ell).
\end{align*}
\end{allowdisplaybreaks}%
Since $u \in S$ and $Z \subseteq S$ above are arbitrary, we have established $\lift(k \kc \ell) = (\lift k) \gc (\lift\ell)$. It remains to consider the case of $\sWhileDo P k$. We put $\phi = \lift k: S \gto S$, and we recall the definitions for the semantic iteration operations:
\begin{align*}
\sWhileDo P k &= \tsum_{\kappa \in \Ord} V_n
&
\gWhileDo P \phi &= \tbigcap_{\kappa \in \Ord} W_\kappa
\\
V_0 &= P \lb 0_S,1_S \rb
&
W_0 &= P \glb \zero_S,\one_S \grb
\\
V_{\kappa+1} &= P \lb k \kc V_\kappa,1_S \rb
&
W_{\kappa+1} &= P \glb \phi \gc W_\kappa, \one_S \grb
\\
V_\lambda &= \tsum_{\kappa<\lambda} V_\kappa,
\ \text{limit ordinal $\lambda$}
&
W_\lambda &= \tbigcap_{\kappa<\lambda} W_\kappa,
\ \text{limit ordinal $\lambda$}
\end{align*}
It is a well-known fact that $\sWhileDo P k = V_\omega$, which says that the least fixpoint closes at $\omega$ iterations. The crucial observation now is that
\[
  W_\kappa = \lift V_\kappa
  \ \text{for every ordinal $\kappa$}.
\]
This is shown by transfinite induction on ordinals. The proof involves using the commutation results for $\lift$ (for 0, 1, conditionals, composition) that we have shown so far. Finally,
\begin{align*}
(u,Y) \in \lift(\sWhileDo P k) &\iff
(\sWhileDo P k)(u) \subseteq Y
\\ &\iff
(\tsum_\kappa V_\kappa)(u) =
\tbigcup_\kappa V_\kappa(u) \subseteq Y
\\ &\iff
V_\kappa(u) \subseteq Y\ \text{for every ordinal $\kappa$}
\\ &\iff
(u,Y) \in \lift V_\kappa = W_\kappa\ \text{for every ordinal $\kappa$}
\\ &\iff
(u,Y) \in \tbigcap_\kappa W_\kappa = \gWhileDo P \phi.
\end{align*}
We have thus shown that $\lift(\sWhileDo P k) = \gWhileDo P (\lift k)$ and the proof is complete.
\end{proof}

Essentially, the above lemma says that the game function operations are a generalization of the nondeterministic function operations. It is an easy exercise to show that the map $\lift$ is injective. So, the algebra $S \nto S$ with the operations of Figure~\ref{fig:nondetOps} is embedded via $\lift$ into the algebra $S \gto S$ with the operations of Figure~\ref{fig:gameOps}.

\begin{defi}[The Implementation Relation]
\label{def:implementation}
Let $k: S \nto S$ be a nondeterministic function and $\phi: S \gto S$ be a game function. We say that \emph{$k$ implements $\phi$} if $\lift k \subseteq f$, and we denote this by $k \simpl \phi$. The definition is meant to capture the idea that $k$ resolves (in some possible way) the angelic nondeterminism of $\phi$. To put it differently, the function $k$ \emph{chooses} for every start state $u$ an output predicate $k(u) \in \phi(u)$ that the angel can guarantee.
\end{defi}

\begin{lemma}[The Implementation Calculus]
\label{lemma:implementation}
The relation $\simpl$ satisfies the following rules:
\begin{gather*}
1_A \simpl \one_A
\qquad
0_{AB} \simpl \zero_{AB}
\qquad
\AxiomC{$P \subseteq S$}
\AxiomC{$k \simpl \phi$}
\AxiomC{$\ell \simpl \psi$}
\TernaryInfC{$P \lb k,\ell \rb \simpl P \glb \phi,\psi \grb$}
\DisplayProof
\qquad
\AxiomC{$k \simpl \phi$}
\AxiomC{$\ell \simpl \psi$}
\BinaryInfC{$k \kc \ell \simpl \phi \gc \psi$}
\DisplayProof
\\[0.5ex]
\AxiomC{$k \simpl \phi$}
\UnaryInfC{$k \simpl \phi \gang \psi$}
\DisplayProof
\qquad
\AxiomC{$\ell \simpl \psi$}
\UnaryInfC{$\ell \simpl \phi \gang \psi$}
\DisplayProof
\qquad
\AxiomC{$k \simpl \phi$\hspace{-1em}}
\AxiomC{$\ell \simpl \psi$}
\BinaryInfC{$k \plus \ell \simpl \phi \gdem \psi$}
\DisplayProof
\qquad
\AxiomC{$P \subseteq S$}
\AxiomC{$k \simpl \phi$}
\BinaryInfC{$\sWhileDo P k \simpl \gWhileDo P \phi$}
\DisplayProof
\end{gather*}
where $k, \ell: S \nto S$ are nondeterministic functions and $\phi, \psi: S \gto S$ are game functions.
\end{lemma}
\begin{proof}
First, we note that all the operations on game functions are monotone w.r.t.\ inclusion. That is, if $\phi \subseteq \phi'$ and $\psi \subseteq \psi'$ then we also have:
\begin{align*}
\phi \gc \psi &\subseteq
\phi' \gc \psi'
&
\phi \gang \psi &\subseteq \phi' \gang \psi'
&
\gWhileDo P \phi &\subseteq \gWhileDo P \phi'
\\
P \glb \phi,\psi \grb &\subseteq
P \glb \phi',\psi' \grb
&
\phi \gdem \psi &\subseteq \phi' \gdem \psi'
\end{align*}
Assume now that $k \simpl \phi$ and $\ell \simpl \psi$, i.e., $\lift k \subseteq \phi$ and $\lift\ell \subseteq \psi$. We obtain the inclusions
\begin{gather*}
\begin{aligned}[t]
\lift 1_S &=
\one_S \subseteq \one_S
\\
\lift 0_S &=
\zero_S \subseteq \zero_S
\\
\lift k &\subseteq
\phi \subseteq
\phi \cup \psi =
\phi \gang \psi
\\
\lift\ell &\subseteq
\psi \subseteq
\phi \cup \psi =
\phi \gang \psi
\end{aligned}
\hspace{3em}
\begin{aligned}[t]
\lift(k \kc \ell) &=
(\lift k) \gc (\lift \ell) \subseteq
\phi \gc \psi
\\
\lift (P \lb k,\ell \rb) &=
P \glb \lift k, \lift\ell \grb \subseteq
P \glb \phi,\psi \grb
\\
\lift (\sWhileDo P k) &=
\gWhileDo{P}{(\lift k)} \subseteq
\gWhileDo{P}{\phi}
\\
\lift (k \plus \ell) &=
(\lift k) \gdem (\lift\ell) \subseteq
\phi \gdem \psi
\end{aligned}
\end{gather*}
using the monotonicity properties for game function operations and the fact that the lifting operation commutes with the semantic program operations (Lemma~\ref{lemma:lift}).
\end{proof}

\begin{defi}[Game Interpretation]
\label{def:gameI}
As in the case of nondeterministic program schemes (Definition~\ref{def:nondetI}), an interpretation of the language of while game schemes consists of a nonempty \emph{state space} $S$ and an \emph{interpretation function} $I$. For a program term $f$, its \emph{interpretation} $I(f): S \gto S$ is a game function on $S$. The function $I$ specifies the meaning of every atomic test, and extends to all tests in the obvious way. Moreover, $I$ specifies the meaning $I(a): S \gto S$ of every atomic action. It extends to all game schemes as:
\begin{gather*}
\begin{aligned}
I(\id) &= \one_S
&
I(f;g) &= I(f) \gc I(g)
&
I(f \ang g) &= I(f) \gang I(g)
&
I(p[f,g]) &= I(p) \glb I(f), I(g) \grb
\\
I(\bot) &= \zero_S
&&&
I(f \dem g) &= I(f) \gdem I(g)
&
I(\wh p f) &= \gWhileDo{I(p)}{I(f)}
\end{aligned}
\end{gather*}
We say that the game interpretation $I$ \emph{lifts} the nondeterministic interpretation $R$ if they have the same state space, and additionally:
\begin{enumerate}[label=(\roman*)]
\item
$I(p) = R(p)$ for every atomic test $p$, and
\item
$I(a) = \lift R(a)$ for every atomic program $a$.
\end{enumerate}
We also say that \emph{$I$ is the lifting of $R$}.
\end{defi}

\begin{defi}[Chain Property]
A decreasing chain of predicates is a transfinite sequence $(X_\kappa)_{\kappa \in \Ord}$ with $X_\kappa \supseteq X_\lambda$ for ordinals $\kappa \leq \lambda$. Let $\phi: S \gto S$ be a game function. We say that $\phi$ satisfies the \emph{chain property} if for every state $u \in S$ and every decreasing chain $(Y_\kappa)_\kappa$ of predicates on $S$, $(u,Y_\kappa) \in \phi$ for all $\kappa$ implies that $(u,\tbigcap_\kappa Y_\kappa) \in \phi$.
\end{defi}

\begin{lemma}[Preservation of Chain Property]
\label{lemma:chain}
The following hold:
\begin{enumerate}
\item
Every non-angelic game function satisfies the chain property.
\item
The game functions $\zero_S$ and $\one_S$ satisfy the chain property.
\item
If the game functions $\phi, \psi: S \gto S$ satisfy the chain property, then so do the game functions $P \glb \phi,\psi \grb$, $\phi \gc \psi$, $\phi \gang \psi$, $\phi \gdem \psi$, and $\gWhileDo P \phi$, where $P$ is a predicate on $S$.
\end{enumerate}
\end{lemma}
\begin{proof}
The most interesting parts of the proof are showing that the operations of angelic choice and composition preserve the chain property. We omit the rest of the proof, since the reader can easily reconstruct it.

For the case $\phi \gang \psi$ of angelic choice, assume that $(u,Y_\kappa) \in \phi \gang \psi$ for every ordinal $\kappa$. We recall the definition $\phi \gang \psi = \phi \cup \psi$, which means that $(u,Y_\kappa) \in \phi$ or $(u,Y_\kappa) \in \psi$ for all $\kappa$. Define the classes $O(\phi)$ and $O(\psi)$ of ordinals as follows:
\begin{align*}
O(\phi) &= \{ \lambda \in \Ord \mid (u,Y_\lambda) \in \phi \}
&
O(\psi) &= \{ \mu \in \Ord \mid (u,Y_\mu) \in \psi \}
\end{align*}
Clearly, the equality $O(\phi) \cup O(\psi) = \Ord$ holds. This implies that at least one of the classes $O(\phi)$, $O(\psi)$ has no upper bound. By symmetry, we only consider the case where $O(\phi)$ has no upper bound, that is: for every ordinal $\kappa$ there is some $\lambda \geq \kappa$ with $\lambda \in O(\phi)$. We extend the subsequence $(Y_\lambda)_{\lambda \in O(\phi)}$ into a decreasing chain $(\hat Y_\lambda)_{\lambda \in \Ord}$ as:
\[
  \hat Y_\lambda = Y_{\lambda'},
  \ \text{where $\lambda' = \text{least} \{ \kappa \in \Ord \mid \text{$\kappa \geq \lambda$ and $\kappa \in O(\phi)$} \}$}.
\]
In particular, if $\lambda \in O(\phi)$ then $\hat Y_\lambda = Y_\lambda$. It is straightforward to verify that $(\hat Y_\lambda)_{\lambda \in \Ord}$ is a decreasing chain with $(u,\hat Y_\lambda) \in \phi$ for every $\lambda \in \Ord$. Since $\phi$ satisfies the chain property, we get that $(u,\tbigcap_{\lambda \in \Ord} \hat Y_\lambda) \in \phi$. Finally, we observe that
\[
  \tbigcap_{\kappa \in \Ord} Y_\kappa =
  \tbigcap_{\lambda \in O(f)} Y_\lambda =
  \tbigcap_{\lambda \in \Ord} \hat Y_\lambda.
\]
This gives us the desired $(u,\tbigcap_{\kappa \in \Ord} Y_\kappa) \in \phi \subseteq \phi \cup \psi$. So, $\phi \gang \psi$ satisfies the chain property.

For the case $\phi \gc \psi$ of composition, we consider the decreasing chain $(Z_\kappa)_\kappa$ and we assume that $(u,Z_\kappa) \in (\phi \gc \psi)$ for all $\kappa$. For every ordinal $\kappa$, define the collection of predicates
\[
  \mathcal Y_\kappa = \{
    Y \subseteq S \mid
    \text{$(u,Y) \in \phi$ and $(v,Z_\kappa) \in \psi$ for all $v \in Y$}.
  \}
\]
The assumption $(u,Z_\kappa) \in (\phi \gc \psi)$ means that the collection $\mathcal Y_\kappa$ is nonempty. We then define the predicate $Y_\kappa = \tbigcup \mathcal Y_\kappa$ and we observe that $Y_\kappa \in \mathcal Y_\kappa$, that is:
\[
  \text{$(u,Y_\kappa) \in \phi$ \qquad and \qquad $(v,Z_\kappa) \in \psi$ for all $v \in Y_\kappa$}.
\]
Moreover, the implications $\kappa \leq \lambda \Imp Z_\kappa \supseteq Z_\lambda \Imp \mathcal Y_\kappa \supseteq \mathcal Y_\lambda \Imp Y_\kappa \supseteq Y_\lambda$ hold. This means that the sequence $(Y_\kappa)_\kappa$ is a decreasing chain. The third containment is justified as follows:
\begin{align*}
Y \in \mathcal Y_\lambda &\implies
\text{$(u,Y) \in \phi$ and $(v,Z_\lambda) \in \psi$ for all $v \in Y$}
\\ &\implies
\text{$(u,Y) \in \phi$ and $(v,Z_\kappa) \in \psi$ for all $v \in Y$}
\\ &\implies
Y \in \mathcal Y_\kappa.
\end{align*}
Since $\phi$ satisfies the chain property, we obtain that $(u,\tbigcap_\kappa Y_\kappa) \in \phi$. Let us consider now an arbitrary element $v$ of $\tbigcap_\kappa Y_\kappa$. We get that $v \in Y_\kappa$ and hence $(v,Z_\kappa) \in \psi$ for every ordinal $\kappa$. But $\psi$ also satisfies the chain property, which gives us that $(v,\tbigcap_\kappa Z_\kappa) \in \psi$. We know that:
\[
  \text{$(u,\tbigcap_\kappa Y_\kappa) \in \phi$ \qquad and \qquad $(v,\tbigcap_\kappa Z_\kappa) \in \psi$ for all $v \in \tbigcap_\kappa Y_\kappa$}.
\]
This means that $(u,\tbigcap_\kappa Z_\kappa) \in (\phi \gc \psi)$. We conclude that $\phi \gc \psi$ satisfies the chain property.
\end{proof}

\begin{theorem}[Full Abstraction]
\label{thm:semantics}
Let $I$ be an interpretation of atomic tests as unary predicates on a state space $S$ and of atomic actions as game functions $S \gto S$ that satisfy the chain property. Then, for every while game scheme $f$, state $u \in S$ and predicate $Y \subseteq S$ we have that: $(u,Y) \in I(f)$ iff Player $\exists$ (the angel) has a winning strategy from the vertex $(u,f)$ in the safety game $G_I(f,\compl Y)$ (recall Definition~\ref{def:opModel}).
\end{theorem}
\begin{proof}
The proof is by induction on the structure of $f$.

First, we consider the case of the atomic action $a$. Recall that we have $C(a) = \{ a,\id \}$. The start vertex for the game is $(u,a)$. The angel has a winning strategy from $(u,a)$ iff there exists some predicate $X$ such that $(u,X) \in I(a)$ and $X \subseteq Y$.
\[
  (u,a) \to (X,\id) \to (v,\id),\ \text{where $v \in X$}
\]
For the case of the skip program $\id$, we have that $C(\id) = \{ \id \}$. The start vertex for the game is $(u,\id)$, and it is also a terminal vertex. So, the angel has a winning strategy in the game $G_I(\id,\compl Y)$ iff $u \in Y$ iff $(u,Y) \in I(\id) = \one_S$.

We handle now the case of the conditional $p[f,g]$. We have that $C(p[f,g]) = \{ p[f,g] \} \cup C(f) \cup C(g)$. Consider a pair $(u,Y)$, where $u \in I(p)$. The case where $u \in I(\neg p)$ is analogous, and we omit it. Notice that there exists a unique transition $(u,p[f,g]) \to (u,f)$. This means that after the transition is taken, any play in $G_I(p[f,g],\compl Y)$ is the same as a play in the game $G_I(f,\compl Y)$. So, we obtain the equivalences:
\begin{align*}
&(u,X) \in I(p[f,g]) \iff (u,X) \in I(f) \iff
\\
&\text{The angel has a winning strategy from $(u,f)$ in $G_I(f,\compl Y)$} \iff
\\
&\text{The angel has a winning strategy from $(u,p[f,g])$ in $G_I(p[f,g],\compl Y)$}.
\end{align*}
The cases $f \ang g$ and $f \dem g$ are handled using similar arguments to the ones we used for the conditional $p[f,g]$, and we therefore omit them.

We will prove now the claim for the while loop $\wh p f$. Recall that $C(\wh p f) = \{ \wh p f,\id \} \cup C(f)@\wh p f$ and $I(\wh p f) = \bigcap_{\kappa \in \Ord} W_\kappa$, where the transfinite sequence $W_\kappa$ is given by
\begin{align*}
W_0 &\triangleq
I(p) \glb \zero_S, \one_S \grb
&
W_{\kappa+1} &\triangleq
I(p) \glb I(f) \gc W_\kappa, \one_S \grb
&
W_\lambda &\triangleq \tbigcap_{\kappa < \lambda} W_\kappa,\   \text{limit ordinal $\lambda$}
\end{align*}
Consider the predicate $Y \subseteq S$, and define the transfinite sequence $(X_\kappa)_{\kappa \in \Ord}$ as follows:
\begin{align*}
X_0 &=
I(p) \cup (\compl I(p) \cap Y)
\\
X_{\kappa+1} &=
\{ u \in S \mid \text{$u \in I(p)$ and $(u,X_\kappa) \in I(f)$} \}
\cup
(\compl I(p) \cap Y)
\\
X_\lambda &= \tbigcap_{\kappa<\lambda} X_\kappa,
\ \text{for limit ordinal $\lambda$}
\end{align*}
The sequence $(X_\kappa)_\kappa$ can be defined equivalently in terms of the approximants $W_\kappa$, as the claim below states. We also put $X = \tbigcap_{\kappa \in \Ord} X_\kappa$. A transfinite induction on $\kappa$ establishes:

\begin{claim*}
$X_\kappa = \{ u \in S \mid (u,Y) \in W_\kappa \}$ for every ordinal $\kappa$. \qed
\end{claim*}

\noindent
The above claim implies in particular that
\[
  X = \{ u \in S \mid (u,Y) \in I(\wh p f) \}.
\]
Moreover, we see below that $X$ is an ``inductive invariant'' for the while loop $\wh p f$.

\begin{claim*}
If $u \in I(p)$ and $u \in X$, then $(u,X)$ is in $I(f)$.
\end{claim*}
\begin{proof}
Suppose that $u \in I(p)$ and $u \in X$, which implies that $u \in X_{\kappa+1}$ for every $\kappa$. From the inductive definition of $X_\kappa$, we obtain that $(u,X_\kappa) \in I(f)$ for every $\kappa$. Since every interpretation $I(a)$ for atomic action $a$ satisfies the chain property, we obtain from Lemma~\ref{lemma:chain} that $I(f)$ satisfies the chain property. It follows that $(u,X) \in I(f)$.
\end{proof}

\noindent
Let us consider now the game $G_I(\wh p f,\compl Y)$.
\begin{itemize}[label=$-$]
\item
Consider a state $u \in I(p)$ with $u \in X$. The previous claim says that $(u,X) \in I(f)$, and hence the I.H.\ gives us that the angel has a winning strategy $\sigma_u$ in the game $G_I(f,\compl X)$. We define the $\exists$-strategy $\sigma$ in the game $G_I(\wh p f,\compl Y)$ as follows: every time a vertex $(u,\wh p f)$ with $u \in I(p)$ is encountered, start playing according to $\sigma_u$. Notice that we have the transition $(u,\wh p f) \to (u,f@\wh p f)$, which means that $\sigma$ simulates $\sigma_u$ on $G_I(f,\compl X)$.

It follows that when the angel plays according to $\sigma$ in the game $G_I(\wh p f,\compl Y)$ with start vertex $(u,\wh p f)$ where $u \in X$, the play will never hit an error vertex in $\compl Y \times \{\id\}$. In particular, if $(u,Y) \in I(\wh p f)$ then $u \in X$ and hence the angel has a winning strategy from $(u,\wh p f)$ in the game $G_I(\wh p f,\compl Y)$.
\item
Let $U$ be the set of states $u \in S$ for which the angel has a winning strategy from $(u,\wh p f)$ in the game $G_I(\wh p f,\compl Y)$. Let $\sigma$ be the (w.l.o.g.\ memoryless, see Theorem~\ref{thm:determined}) strategy of Player $\exists$ that witnesses his winning region in the game $G_I(\wh p f,\compl Y)$.

Consider a state $u \in I(p)$ with $u \in U$. If the angel plays according to $\sigma$ in the game $G_I(f,\compl U)$, then he wins, because $\sigma$ keeps the play within the winning region. The I.H.\ then says that $(u,U) \in I(f)$.\enlargethispage{\baselineskip}
\begin{claim}
$U \subseteq X$.
\end{claim}
\begin{proof}
It suffices to show that $U \subseteq X_\kappa$ for every ordinal $\kappa$. For the base case $\kappa = 0$, the claim $U \subseteq X_0 = I(p) \cup (\compl I(p) \cap Y)$ is obvious. For successor ordinals:
\begin{align*}
X_{\kappa+1} &=
\{ u \in S \mid \text{$u \in I(p)$ and $(u,X_\kappa) \in I(f)$} \}
\cup
(\compl I(p) \cap Y)
\\
&\supseteq
\{ u \in S \mid \text{$u \in I(p)$ and $(u,U) \in I(f)$} \}
\cup
(\compl I(p) \cap Y)
\\
&\supseteq
\{ u \in S \mid \text{$u \in I(p)$ and $u \in U$} \}
\cup
(\compl I(p) \cap Y),
\end{align*}
which is equal to $U$. The case of limit ordinals is easy.
\end{proof}

\noindent
Suppose now that $(u,\wh p f)$ is in the winning region of the angel in the game $G_I(\wh p f,\compl Y)$. It follows that $u \in U$ and hence $u \in X$. We thus conclude that $(u,Y)$ is in $I(\wh p f)$.
\end{itemize}
This completes the proof for the case of the while loop $\wh p f$.

Finally, we have to deal with the case $e;f$ of sequential composition. Recall the definitions $C(e;f) = C(e)@f \cup C(f)$ and $I(e;f) = I(e) \gc I(f)$.
\begin{itemize}[label=$-$]
\item
Suppose that $(u,Z) \in I(e;f)$. There exists $Y \subseteq S$ with $(u,Y) \in I(e)$ and $(v,Z) \in I(f)$ for every $v \in Y$. The I.H.\ says that there exists a winning $\exists$-strategy $\sigma$ for the game $G_I(e,\compl Y)$ started at vertex $(u,e)$. Moreover, for every $v \in Y$, there exists a winning $\exists$-strategy $\tau_v$ for the game $G_I(f,\compl Z)$ started at vertex $(v,f)$. Now, we define the strategy $\rho$ for the game $G_I(e;f,\compl Y)$ as follows: start playing according to $\sigma$, and as soon as you encounter a vertex $(v,f)$ start playing according to $\tau_v$. The $\exists$-strategy $\rho$ is winning for the angel in the game $G_I(e;f,\compl Z)$ when started at $(u,e;f)$.
\item
Suppose now that the angel has a (w.l.o.g.\ memoryless, see Theorem~\ref{thm:determined}) winning strategy $\rho$ from the vertex $(u,e;f)$ in the game $G_I(e;f,\compl Z)$. Let
\begin{align*}
Y = \{ v \in S \mid {}
&\text{the vertex $(v,\id;f)$ appears in some $\rho$-play starting from $(u,e;f)$} \}.
\end{align*}
Then, the angel has a winning strategy from $(u,e)$ in the game $G_I(e,\compl Y)$. Moreover, for every $v \in Y$, the angel has a winning strategy from $(v,f)$ in the game $G_I(f,\compl Z)$. From the I.H., it follows that $(u,Y) \in I(e)$. Moreover, for every $v \in Y$, we obtain that $(v,Z) \in I(f)$. So, $(u,Z) \in I(e;f)$.
\end{itemize}
This concludes the argument for the case of composition, and the proof is thus complete.
\end{proof}

\section{A Hoare Calculus for While Game Schemes}
\label{sec:hoare}

In this section, we present formulas that are used to specify programs. The basic formulas are Hoare assertions of the form $\hoare{p}f{q}$, and we also consider assertions under certain hypotheses $\Phi, \Psi$ of a simple form. The latter formulas are called Hoare implications and are of the form $\Phi,\Psi \Imp \hoare{p}f{q}$. We will then continue to present our first axiomatization, with which we derive valid Hoare implications.

\begin{defi}[Tests and Entailment]
Let $I$ be an interpretation of the atomic tests, which extends to all tests in the obvious way. For a test $p$ and a state $u \in S$, we write $I,u \models p$ when $u \in I(p)$. We read this as: ``the state $u$ satisfies $p$ (under $I$)''. When $I,u \models p$ for every state $u \in S$, we say that $I$ \emph{satisfies} $p$, and we write $I \models p$. For a set $\Phi$ of tests, the interpretation $I$ \emph{satisfies} $\Phi$ if it satisfies every test in $\Phi$. We then write $I \models \Phi$. Finally, we say that $\Phi$ \emph{entails} $p$, denoted $\Phi \models p$, if $I \models \Phi$ implies $I \models p$ for every $I$.
\end{defi}

\begin{defi}[Hoare Assertions]
\label{def:simple}
An expression $\hoare p f q$, where $p$ and $q$ are tests and $f$ is a program term, is called a \emph{Hoare assertion}. The test $p$ is called the \emph{precondition} and the test $q$ is called the \emph{postcondition} of the assertion. Informally, the formula $\hoare p f q$ says that when the program $f$ starts at a state satisfying the predicate $p$, then the angel has a strategy so that whatever the demon does, the final state (upon termination) satisfies the predicate $q$. The Hoare assertion $\hoare p a q$, where $a$ is an atomic program, is called a \emph{simple Hoare assertion}. More formally, we say that the interpretation $I$ \emph{satisfies} $\hoare p f q$ when
\[
  \text{$I,u \models p$ implies that $(u,I(q)) \in I(f)$}
\]
for every state $u \in S$. We then write $I \models \hoare p f q$.
\end{defi}

\begin{defi}[Simple Hoare Implications \& Weak Hoare Theory]
\label{def:implication}
Let $\Phi$ be a finite set of tests, and $\Psi$ be a finite set of simple Hoare assertions. We call the expression
\[
  \Phi,\Psi \Imp \hoare{p}f{q}
\]
a \emph{simple Hoare implication}. The tests in $\Phi$ and the simple assertions in $\Psi$ are the \emph{hypotheses} of the implication, and the Hoare assertion $\hoare p f q$ is the \emph{conclusion}. We use the qualifier \emph{simple} for implications of  the form $\Phi,\Psi \Imp \hoare p f q$, because the hypotheses $\Psi$ involve only simple Hoare assertions (instead of general Hoare assertions for arbitrary programs).

Let $I$ be an interpretation of tests and actions. We say that $I$ \emph{satisfies} the implication $\Phi,\Psi \Imp \hoare p f q$, which we denote by $I \models \Phi,\Psi \Imp \hoare p f q$, when the following holds: If the interpretation $I$ satisfies every test in $\Phi$ and every assertion in $\Psi$, then $I$ satisfies the assertion $\hoare p f q$. An implication $\Phi,\Psi \Imp \hoare{p}f{q}$ is \emph{valid}, denoted $\Phi,\Psi \models \hoare p f q$, if every interpretation satisfies it. The set of all valid Hoare implications forms the \emph{weak Hoare theory} of while game schemes.
\end{defi}

\begin{defi}[Boolean Atoms \& $\Phi$-Consistency]
Suppose that we have fixed a finite set of atomic tests. For an atomic test $p$, the expressions $p$ and $\neg p$ are called \emph{literals} for $p$ (\emph{positive} and \emph{negative} respectively). Fix an enumeration $p_1, p_2, \ldots, p_k$ of the atomic tests. A \emph{Boolean atom} (or simply \emph{atom}) is an expression $\ell_1 \ell_2 \cdots \ell_k$, where every $\ell_i$ is a literal for $p_i$. We use lowercase letters $\alpha, \beta, \gamma, \ldots$ from the beginning of the Greek alphabet to range over atoms. An atom is essentially a conjunction of literals, and it can also be thought of as a propositional truth assignment. We write $\alpha \leq p$ to mean that the atom $\alpha$ satisfies the test $p$. We denote by $\At$ the set of all atoms.

Assume that $\Phi$ is a finite set of tests. We say that an atom $\alpha$ is \emph{$\Phi$-consistent} if $\alpha \leq p$ for every test $p$ in $\Phi$. We write $\AtC$ for the set of all $\Phi$-consistent atoms.
\end{defi}

\begin{defi}[The Free Test Interpretation]
\label{def:freeTestI}
Let $\Phi$ be a finite set of tests. We define the interpretation $I_\Phi$ on tests, which is called the \emph{free test interpretation} w.r.t.\ $\Phi$. The state space is the set $\AtC$ of $\Phi$-consistent atoms, and every test is interpreted as a unary predicate on $\AtC$. For an atomic test $p$, define its interpretation
\[
  I_\Phi(p) \triangleq
  \{ \alpha \in \AtC \mid \alpha \leq p \}
\]
to be the set of $\Phi$-consistent atoms that satisfy $p$. In fact, an easy induction on the structure of tests proves that for every (atomic or composite) test $p$, $I_\Phi(p)$ is equal to the set of $\Phi$-consistent atoms that satisfy $p$.
\end{defi}

\begin{note}[Complete Boolean Calculus]
\label{note:bool}
We assume that we have a complete Boolean calculus, with which we derive judgments $\Phi \vdash p$, where $\Phi$ is a finite set of tests and $p$ is a test. This means that the statements
\begin{align*}
\Phi \models p
&&
I_\Phi \models p
&&
I_\Phi(p) = \AtC
&&
\Phi \vdash p
\end{align*}
are all equivalent. From this we also obtain that $I_\Phi(p) \subseteq I_\Phi(q)$ iff $\Phi \vdash p \to q$.
\end{note}

We propose now a Hoare-style calculus (Figure~\ref{fig:angHL}), which is used for deriving simple Hoare implications that involve while game schemes. As we will show, the calculus of Figure~\ref{fig:angHL} is sound and complete for the weak Hoare theory of while game schemes. Establishing soundness is a relatively straightforward result. The most interesting part is the soundness of the ($\rLoop$) rule for while loops. The observation is that the loop invariant defines a ``safe region'' of the game, and the angel has a strategy to keep a play within this region.

\begin{figure}[t]
\centering
$\begin{gathered}
\AxiomC{$\hoare p a q$ in $\Psi$}
\RightLabel{($\rHyp$)}
\UnaryInfC{$\Phi,\Psi \vdash \hoare{p}a{q}$}
\DisplayProof
\qquad
\AxiomC{\phantom{$\{\Psi$}}
\RightLabel{($\rSkip$)}
\UnaryInfC{$\Phi,\Psi \vdash
\hoare{p}\id{p}$}
\DisplayProof
\qquad
\AxiomC{\phantom{$\{\Psi$}}
\RightLabel{($\rDvrg$)}
\UnaryInfC{$\Phi,\Psi \vdash
\hoare{p}\bot{q}$}
\DisplayProof
\\[1ex]
\AxiomC{$\begin{aligned}
  \Phi,\Psi &\vdash \hoare{p}f{q}
  \\[-0.5ex]
  \Phi,\Psi &\vdash \hoare{q}g{r}
\end{aligned}$}
\RightLabel{($\rSeq$)}
\UnaryInfC{$\Phi,\Psi \vdash \hoare{p}{f;g}{r}$}
\DisplayProof
\qquad
\AxiomC{$\begin{aligned}
  \Phi,\Psi &\vdash \hoare{q \land p}f{r}
  \\[-0.5ex]
  \Phi,\Psi &\vdash \hoare{q \land \neg p}g{r}
\end{aligned}$}
\RightLabel{($\rCond$)}
\UnaryInfC{$\Phi,\Psi \vdash \hoare{q}{\ifThenElse p f g}{r}$}
\DisplayProof
\\[1ex]
\AxiomC{$\Phi,\Psi \vdash \hoare{r \land p}f{r}$}
\RightLabel{($\rLoop$)}
\UnaryInfC{$\Phi,\Psi \vdash \hoare{r}{\whileDo p f}{r \land \neg p}$}
\DisplayProof
\\[1ex]
\AxiomC{$\Phi,\Psi \vdash \hoare{p}{f_i}{q}$}
\RightLabel{($\rAng_i$)}
\UnaryInfC{$\Phi,\Psi \vdash \hoare{p}{f_1 \ang f_2}{q}$}
\DisplayProof
\qquad
\AxiomC{$\Phi,\Psi \vdash \hoare{p}f{q}$}
\AxiomC{$\Phi,\Psi \vdash \hoare{p}g{q}$}
\RightLabel{($\rDem$)}
\BinaryInfC{$\Phi,\Psi \vdash \hoare{p}{f \dem g}{q}$}
\DisplayProof
\\[1ex]
\AxiomC{$\Phi \vdash p' \to p$}
\AxiomC{$\Phi,\Psi \vdash \hoare{p}f{q}$}
\AxiomC{$\Phi \vdash q \to q'$}
\RightLabel{($\rWeak$)}
\TernaryInfC{$\Phi,\Psi \vdash \hoare{p'}f{q'}$}
\DisplayProof
\\[1ex]
\AxiomC{$\Phi,\Psi \vdash \hoare{p_1}f{q}$}
\AxiomC{$\Phi,\Psi \vdash \hoare{p_2}f{q}$}
\RightLabel{($\rJoin$)}
\BinaryInfC{$\Phi,\Psi \vdash \hoare{p_1 \lor p_2}f{q}$}
\DisplayProof
\qquad
\begin{aligned}
\Phi,\Psi &\vdash \hoare{\false}f{q}
\quad(\rJoin_0)
\\
\Phi,\Psi &\vdash \hoare{p}f{\true}
\quad(\rMeet_0)
\end{aligned}
\end{gathered}$
\caption{\emph{Game Hoare Logic}: A sound and complete Hoare-style calculus for while program schemes with angelic and demonic nondeterministic choice.}
\label{fig:angHL}
\end{figure}

\begin{observation}[Variant Rule for Demonic Choice]
We can have a slightly more flexible form of the rule for demonic choice. The following rule is admissible:
\[
  \AxiomC{$\Phi,\Psi \vdash \hoare{p}f{q}$}
  \AxiomC{$\Phi,\Psi \vdash \hoare{p}g{r}$}
  \RightLabel{($\rDem'$).}
  \BinaryInfC{$\Phi,\Psi \vdash \hoare{p}{f \dem g}{q \lor r}$}
  \DisplayProof
\]
The proof that $(\rDem')$ is admissible is straightforward:
\[
  \AxiomC{$\Phi,\Psi \vdash \hoare{p}f{q}$}
  \AxiomC{$\Phi \vdash q \to q \lor r$}
  \RightLabel{($\rWeak$)}
  \BinaryInfC{$\Phi,\Psi \vdash \hoare{p}f{q \lor r}$}
  \AxiomC{$\Phi,\Psi \vdash \hoare{p}g{r}$}
  \AxiomC{$\Phi \vdash r \to q \lor r$}
  \RightLabel{($\rWeak$)}
  \BinaryInfC{$\Phi,\Psi \vdash \hoare{p}g{q \lor r}$}
  \RightLabel{($\rDem$).}
  \BinaryInfC{$\Phi,\Psi \vdash \hoare{p}{f \dem g}{q \lor r}$}
  \DisplayProof
\]
Notice the similarity of the rule $(\rDem')$ with the definition of the semantic demonic choice operation $\gdem$ in Figure~\ref{fig:gameOps}.
\end{observation}

\begin{observation}[Weakening The Trivial Rules]
\label{obs:trivialRules}
In the Hoare-style calculus of Figure~\ref{fig:angHL} we included two ``trivial'' axioms:
\begin{gather*}
\AxiomC{}
\RightLabel{($\rJoin_0$)}
\UnaryInfC{$\Phi,\Psi \vdash \hoare{\false}f{q}$}
\DisplayProof
\qquad
\AxiomC{}
\RightLabel{($\rMeet_0$)}
\UnaryInfC{$\Phi,\Psi \vdash \hoare{p}f{\true}$}
\DisplayProof
\end{gather*}
We claim that they can be weakened into the axioms
\begin{gather*}
\AxiomC{}
\RightLabel{($a$-$\rJoin_0$)}
\UnaryInfC{$\Phi,\Psi \vdash \hoare{\false}a{q}$}
\DisplayProof
\qquad
\AxiomC{}
\RightLabel{($a$-$\rMeet_0$)}
\UnaryInfC{$\Phi,\Psi \vdash \hoare{p}a{\true}$}
\DisplayProof
\end{gather*}
so that they apply only to atomic programs $a, b, \ldots$, without changing the theory generated by the calculus. The claim is that if we replace ($\rJoin_0$) and ($\rMeet_0$) by the weaker axioms ($a$-$\rJoin_0$) and ($a$-$\rMeet_0$), then we can still prove ($\rJoin_0$) and ($\rMeet_0$) for arbitrary terms.
\end{observation}
\begin{proof}
Suppose that $\vdash_w$ denotes provability in the weakened proof system with ($a$-$\rJoin_0$) and ($a$-$\rMeet_0$). We claim that for every program term $f$ and all tests $p, q$, it holds:
\begin{gather*}
\vdash_w \hoare{\false}f{q}
\qquad\text{and}\qquad
\vdash_w \hoare{p}f{\true}.
\end{gather*}
It suffices to establish that $\vdash_w \hoare{\false}f{\false}$ and $\vdash_w \hoare{\true}f{\true}$, because we have:
\begin{gather*}
\AxiomC{$\vdash_w \hoare{\false}f{\false}$}
\AxiomC{$\vdash \false \to q$}
\RightLabel{($\rWeak$)}
\BinaryInfC{$\vdash_w \hoare{\false}f{q}$}
\DisplayProof
\quad\ 
\AxiomC{$\vdash_w \hoare{\true}f{\true}$}
\AxiomC{$\vdash p \to \true$}
\RightLabel{($\rWeak$)}
\BinaryInfC{$\vdash_w \hoare{p}f{\true}$}
\DisplayProof
\end{gather*}
The proof is by induction on the structure of $f$. We will only give the following derivation
\[
  \AxiomC{$\false \land p \to \false$}
  \AxiomC{$\hoare{\false}f{\false}$ (I.H.)}
  \RightLabel{($\rWeak$)}
  \BinaryInfC{$\hoare{\false \land p}f{\false}$}
  \RightLabel{($\rLoop$)}
  \UnaryInfC{$\hoare{\false}{\wh p f}{\false \land \neg p}$}
  \AxiomC{$\false \land \neg p \to \false$}
  \RightLabel{($\rWeak$)}
  \BinaryInfC{$\hoare{\false}{\wh p f}{\false}$}
  \DisplayProof
\]
as an illustrative example. The other cases equally straightforward and we omit them.
\end{proof}

\begin{theorem}[Soundness]
\label{thm:sound}
The Hoare calculus of Figure~\ref{fig:angHL} is sound.
\end{theorem}
\begin{proof}
The soundness of the proposed Hoare calculus is an immediate consequence of the following properties that are formulated at a purely semantic level.
\begin{gather*}
\begin{aligned}
\hoare{P}{&\one_S}{P}
\\[-0.5ex]
\hoare{P}{&\zero_S}{Q}
\end{aligned}
\qquad
\AxiomC{$\hoare{P}\phi{Q}$}
\AxiomC{$\hoare{Q}\psi{R}$}
\BinaryInfC{$\hoare{P}{\phi \gc \psi}{R}$}
\DisplayProof
\qquad
\AxiomC{$\hoare{Q \cap P}\phi{R}$}
\AxiomC{$\hoare{Q \cap \compl P}\psi{R}$}
\BinaryInfC{$\hoare{Q}{P \glb \phi,\psi \grb}{R}$}
\DisplayProof
\\[1ex]
\AxiomC{$\hoare{R \cap P}\phi{R}$}
\UnaryInfC{$\hoare{R}{\gWhileDo P \phi}{R \cap \compl P}$}
\DisplayProof
\qquad
\AxiomC{$\hoare{P}{\phi_i}{Q}$}
\UnaryInfC{$\hoare{P}{\phi_1 \cup \phi_2}{Q}$}
\DisplayProof
\qquad
\AxiomC{$\hoare{P}{\phi_\kappa}{Q}$}
\UnaryInfC{$\hoare{P}{\tbigcap_\kappa \phi_\kappa}{Q}$}
\DisplayProof
\\[1ex]
\AxiomC{$P' \subseteq P$}
\AxiomC{$\hoare{P}{\phi}{Q}$}
\AxiomC{$Q \subseteq Q'$}
\TernaryInfC{$\hoare{P'}{\phi}{Q'}$}
\DisplayProof
\qquad
\AxiomC{$\hoare{P_1}{\phi}{Q}$}
\AxiomC{$\hoare{P_2}{\phi}{Q}$}
\BinaryInfC{$\hoare{P_1 \cup P_2}{\phi}{Q}$}
\DisplayProof
\qquad
\begin{aligned}
\hoare{\emptyset}{&\phi}{Q}
\\[-0.5ex]
\hoare{P}{&\phi}{S}
\end{aligned}
\end{gather*}
For predicates $P, Q \subseteq S$ and a game function $\phi: S \gto S$, we understand $\hoare{P}{\phi}{Q}$ as the assertion saying that $(u,Q) \in \phi$ for every state $u \in P$. Establishing the above semantic properties of game functions is a tedious but straightforward task. We will therefore only consider here the case $\gWhileDo P \phi$ and leave the rest to the reader. Recall the definition:
\begin{align*}
\gWhileDo{P}{\phi}& =
\tbigcap_{\kappa \in \Ord} W_\kappa
&
W_0 &=
P \glb \zero_S,\one_S \grb
\\
W_{\kappa+1} &=
P \glb \phi \gc W_\kappa,\one_S \grb
&
W_\lambda &=
\tbigcap_{\kappa<\lambda} W_\kappa,
\ \text{for limit ordinal $\lambda$}
\end{align*}
We show by transfinite induction that $\hoare{R}{W_\kappa}{R \cap \neg P}$. Indeed, for the base case $W_0$ and for the case of the successor ordinal $W_{\kappa+1}$ we have the following derivations:
\begin{allowdisplaybreaks}
\begin{gather*}
\AxiomC{$\hoare{R \cap P}{\zero_S}{R \cap \compl P}$}
\AxiomC{$\hoare{R \cap \compl P}{\one_S}{R \cap \compl P}$}
\BinaryInfC{$\hoare{R}{P \glb \zero_S,\one_S \grb}{R \cap \compl P}$}
\DisplayProof
\\[1ex]
\AxiomC{$\hoare{R \cap P}{\phi}{R}$ (hyp.)}
\AxiomC{$\hoare{R}{W_\kappa}{R \cap \compl P}$ (I.H.)}
\BinaryInfC{$\hoare{R \cap P}{\phi \gc W_\kappa}{R \cap \compl P}$}
\AxiomC{$\hoare{R \cap \compl P}{\one_S}{R \cap \compl P}$}
\BinaryInfC{$\hoare{R}{P \glb \phi \gc W_\kappa,\one_S \grb}{R \cap \compl P}$}
\DisplayProof
\end{gather*}
\end{allowdisplaybreaks}%
The case $W_\lambda$ of the limit ordinal $\lambda$ is handled using the I.H.\ for each ordinal $\kappa < \lambda$ and the infinitary rule for $\tbigcap$. Finally, the assertion $\hoare{R}{\gWhileDo P \phi}{R \cap \compl P}$ is shown using the claim and the rule for infinitary intersection.
\end{proof}

\begin{example}
\label{ex:proof}
\newcommand{\inv}{\mathit{inv}}
We will use the Hoare logic of Figure~\ref{fig:angHL} to establish the partial-correctness property $\hoare{x=0}{h}{x=1}$ for the program $h$ of Example~\ref{ex:operational} (recall the abbreviations $f, g$).
\begin{gather*}
\begin{aligned}[t]
1.\ &
\hoare{x=0}\id{x=0}
&&\text{[$\rSkip$]}
\\
2.\ &
\hoare{x=0}{\id \ang \inc}{x=0}
&&\text{[1, $\rAng$]}
\\
3.\ &
\hoare{x=0}{\inc}{x=1}
&&\text{[hypothesis]}
\\
4.\ &
\hoare{x=0}{\id \dem \inc}{\inv}
&&\text{[1, 3, $\rDem'$]}
\\
5.\ &
\hoare{x=0}{f;g}{\inv}
&&\text{[2, 4, $\rSeq$]}
\\
6.\ &
\hoare{\inv \land (x=0)}{f;g}{\inv}
&&\text{[5, bool, $\rWeak$]}
\\
7.\ &
\hoare{\inv}{h}{\inv \land (x \neq 0)}
&&\text{[6, $\rLoop$]}
\\
8.\ &
\hoare{x=0}{h}{x=1}
&&\text{[7, bool, $\rWeak$]}
\end{aligned}
\hspace{2.5em}
\begin{aligned}[t]
&\{ \mathsf{Precondition}: x=0 \}
\\
&\text{// invariant $\inv \triangleq (x=0) \lor (x=1)$}
\\
&\kwWhile\ (x=0)\ \kwDo
\\
&\qquad \text{// $\inv \land (x=0) \lrto (x=0)$}
\\
&\qquad \id \ang \inc
\\
&\qquad \id \dem \inc
\\
&\qquad \text{// $(x=0) \lor (x=1)$}
\\
&\{ \mathsf{Postcondition}: x=1 \}
\end{aligned}
\end{gather*}
The only hypothesis for atomic symbols used in the proof is $\hoare{x=0}\inc{x=1}$.
\end{example}

\section{First Completeness Theorem: Weak Hoare Theory}
\label{sec:completeA}

We will now prove the completeness of the Hoare calculus of Figure~\ref{fig:angHL} with respect to the class of all interpretations. This means that we consider arbitrary interpretations of the atomic programs $a, b, \ldots$ as game functions. So, the deductive system of Figure~\ref{fig:angHL} is complete for the weak Hoare theory of while game schemes. Note that this is an \emph{unconditional} completeness result (no extra assumptions about expressiveness or about the first-order theory of the domain of computation), not a relative completeness theorem in the sense of \cite{cook1978}.

We show our result by constructing a ``free'' interpretation $\gI$ from the hypotheses $\Phi$ and $\Psi$ about the atomic symbols. We can think of this interpretation as the \emph{least restrictive} interpretation that satisfies the hypotheses. Completeness follows from the fact that the interpretation $\gI$ characterizes the theory generated by our calculus. In other words, everything that is true in $\gI$ is provable using our partial-correctness calculus.

\begin{defi}[The Free Game Interpretation]
Let $\Phi$ be a finite set of tests, and $\Psi$ be a finite set of simple Hoare assertions. We define the \emph{free game interpretation} $\gI$ (w.r.t.\ $\Phi$ and $\Psi$) to have $\AtC$ as state space, and to interpret the tests as $I_\Phi$ (the free test interpretation w.r.t.\ $\Phi$, see Definition~\ref{def:freeTestI}) does. Moreover, the interpretation $\gI(a): \AtC \gto \AtC$  of the atomic action $a$ is given by: for every $\Phi$-consistent atom $\alpha$,
\begin{itemize}[label=$-$]
\item
$(\alpha,\AtC) \in \gI(a)$, and for every subset $X \subsetneq \AtC$,
\item
$(\alpha,X) \in \gI(a)$ iff there exists $\hoare p a q \in \Psi$ s.t.\ $\alpha \leq p$ and $I_\Phi(q) \subseteq X$.
\end{itemize}
\end{defi}

\begin{lemma}
\label{lemma:free}
Let $\Phi$ be a finite set of tests, and $\Psi$ be a finite set of simple Hoare assertions. The free game interpretation $\gI$ satisfies all formulas in $\Phi$ and $\Psi$. \qed
\end{lemma}

\begin{theorem}[Completeness]
\label{thm:complete}
Let $\Phi$ be a finite set of tests, and $\Psi$ be a finite set of simple Hoare assertions. For every program term $f$ and every $\Phi$-consistent atom $\alpha$,
\[
  \text{$(\alpha,X) \in \gI(f)$ implies that $\Phi,\Psi \vdash \hoare{\alpha}f{\tbigvee X}$}.
\]
\end{theorem}
\begin{proof}
The proof proceeds by induction on the structure of the program term $f$. Recall that we have assumed having a complete Boolean calculus (see Note~\ref{note:bool}).

We begin with the base case of the skip program $\id$. Consider an arbitrary pair $(\alpha,X)$ of $\gI(\id)$, where $\alpha \in X$. Since $\gI(\id) = \one_{\AtC}$, we know that $\alpha \in X$. Using the ($\rSkip$) axiom and the weakening rule, we have the derivation:
\[
  \AxiomC{}
  \RightLabel{$(\rSkip)$}
  \UnaryInfC{$\Phi,\Psi \vdash \hoare{\alpha}\id{\alpha}$}
  \AxiomC{$\alpha \in X \subseteq \AtC$}
  \UnaryInfC{$\Phi \vdash \alpha \to \bigvee X$}
  \RightLabel{$(\rWeak)$.}
  \BinaryInfC{$\Phi,\Psi \vdash \hoare{\alpha}\id{\bigvee X}$}
  \DisplayProof
\]
Now, we handle the case of the always diverging program $\bot$. Let $(\alpha,X)$ be an arbitrary element of $\gI(\bot) = \zero_\AtC = \AtC \times \wp\AtC$. The ($\rDvrg$) axiom gives us immediately
\[
  \AxiomC{}
  \RightLabel{($\rDvrg$).}
  \UnaryInfC{$\Phi,\Psi \vdash \hoare{\alpha}\bot{\bigvee X}$}
  \DisplayProof
\]
For the case of an atomic action $a$, consider an arbitrary pair $(\alpha,X)$ in $\gI(a)$. If $X = \AtC$, then we have the following derivation:
\[
  \AxiomC{}
  \RightLabel{($\rMeet_0$)}
  \UnaryInfC{$\Phi,\Psi \vdash \hoare{\alpha}a{\true}$}
  \AxiomC{$I_\Phi(\true) = \AtC$}
  \UnaryInfC{$\Phi \vdash \true \to \bigvee\AtC$}
  \RightLabel{($\rWeak$).}
  \BinaryInfC{$\Phi,\Psi \vdash \hoare{\alpha}a{\bigvee\AtC}$}
  \DisplayProof
\]
Assume now that $X \subsetneq \AtC$. By definition of $\gI(a)$, there exists a simple Hoare hypothesis $\hoare{p}a{q}$ in $\Psi$ such that $\alpha \leq p$ and $I_\Phi(q) \subseteq X$. So,
\[
  \AxiomC{$\alpha \leq p$}
  \UnaryInfC{$\Phi \vdash \alpha \to p$}
  \AxiomC{$\hoare p a q$ in $\Psi$}
  \RightLabel{($\rHyp$)}
  \UnaryInfC{$\Phi,\Psi \vdash \hoare p a q$}
  \AxiomC{$I_\Phi(q) \subseteq X \subseteq \AtC$}
  \UnaryInfC{$\Phi \vdash q \to \bigvee X$}
  \RightLabel{($\rWeak$).}
  \TernaryInfC{$\Phi,\Psi \vdash \hoare{\alpha}a{\bigvee X}$}
  \DisplayProof
\]
This concludes the proof for the case of atomic programs.

We will handle now the case $f;g$ of sequential composition. Let $(\alpha,Y)$ be an arbitrary pair in $\gI(f;g) = \gI(f) \gc \gI(g)$. By definition of the $\gc$ operation on game functions, there exists $X \subseteq \AtC$ such that $(\alpha,X) \in \gI(f)$, and $(\beta,Y) \in \gI(g)$ for every $\beta \in X$. So,
\[
  \AxiomC{$(\alpha,X)$ in $\gI(f)$}
  \RightLabel{(I.H.)}
  \UnaryInfC{$\Phi,\Psi \vdash \hoare{\alpha}f{\bigvee X}$}
  \AxiomC{$(\beta,Y)$ in $\gI(g)$}
  \RightLabel{(I.H.)}
  \UnaryInfC{$\Phi,\Psi \vdash \hoare{\beta}g{\bigvee Y}$}
  \AxiomC{$\beta \in X$}
  \RightLabel{($\rJoin$)}
  \BinaryInfC{$\Phi,\Psi \vdash \hoare{\bigvee X}g{\bigvee Y}$}
  \RightLabel{($\rSeq$).}
  \BinaryInfC{$\Phi,\Psi \vdash \hoare{\alpha}{f;g}{\bigvee Y}$}
  \DisplayProof
\]
Observe in the derivation above that we may have to apply the ($\rJoin$) rule several times (finitely many), because $X$ may contain several $\Phi$-consistent atoms.

For the case of the conditional $\ifThenElse p f g$, let us consider a pair $(\alpha,X)$ in $\gI(p[f,g])$. We deal with the case where $\alpha \leq p$. We obtain the following derivations:
\begin{allowdisplaybreaks}
\begin{gather*}
\AxiomC{$\Phi \vdash \alpha \land p \to \alpha$}
\AxiomC{$(\alpha,X)$ in $\gI(f)$}
\RightLabel{(I.H.)}
\UnaryInfC{$\Phi,\Psi \vdash \hoare{\alpha}f{\bigvee X}$}
\RightLabel{($\rWeak$)}
\BinaryInfC{(1) $\Phi,\Psi \vdash \hoare{\alpha \land p}f{\bigvee X}$}
\DisplayProof
\\[1ex]
\AxiomC{$\alpha \leq p$}
\UnaryInfC{$\Phi \vdash \alpha \land \neg p \to \false$}
\AxiomC{$$}
\RightLabel{($\rJoin_0$)}
\UnaryInfC{$\Phi,\Psi \vdash \hoare{\false}g{\bigvee X}$}
\RightLabel{($\rWeak$)}
\BinaryInfC{(2) $\Phi,\Psi \vdash \hoare{\alpha \land \neg p}g{\bigvee X}$}
\DisplayProof
\\
\AxiomC{\vdots}
\LeftLabel{(1)}
\UnaryInfC{$\Phi,\Psi \vdash \hoare{\alpha \land p}f{\bigvee X}$}
\AxiomC{\vdots}
\LeftLabel{(2)}
\UnaryInfC{$\Phi,\Psi \vdash \hoare{\alpha \land \neg p}g{\bigvee X}$}
\RightLabel{($\rCond$)}
\BinaryInfC{$\Phi,\Psi \vdash \hoare{\alpha}{\ifThenElse p f g}{\bigvee X}$}
\DisplayProof
\end{gather*}
\end{allowdisplaybreaks}%
The proof for the case where $\alpha \leq \neg p$ is completely analogous.

We handle now the case of the loop $\wh p f$. Let $(\gamma,\Gamma)$ be an arbitrary pair in the game function
$\gI(\wh p f) =
 \gWhileDo{I_\Phi(p)}{\gI(f)} =
 \bigcap_i W_i
$,
where the sequence $W_i$ is given by
\begin{align*}
W_0 &= I_\Phi(p) \glb \zero_S,\one_S \grb
&
W_{i+1} &= I_\Phi(p) \glb \gI(f) \gc W_i,\one_S \grb
\end{align*}
We do not need to consider the entire transfinite sequence $(W_\kappa)_{\kappa \in \Ord}$, because the space of game functions on $\AtC$ is finite and hence the sequence stabilizes in a finite number of steps. Define the sequence $(V_i)_{i \geq 0}$ by
\begin{align*}
V_i = \{ \alpha \in \AtC \mid
  (\alpha,\Gamma) \in W_i
\}.
\end{align*}
We give an inductive definition for $(V_i)_{i \geq 0}$ that is equivalent to the above:
\begin{allowdisplaybreaks}
\begin{align*}
V_0 &=
\{ \alpha \in \AtC \mid (\alpha,\Gamma) \in W_0 \}
\\ &=
\{ \alpha \in \AtC \mid \text{$\alpha \leq p$ or ($\alpha \leq \neg p$ and $\alpha \in \Gamma$}) \}
\\ &=
\{ \alpha \in \AtC \mid \text{$\alpha \leq p$ or $\alpha \in \Gamma$} \}
\\ &=
I_\Phi(p) \cup (\compl I_\Phi(p) \cap \Gamma)
\\
V_{i+1} &=
\{ \alpha \in \AtC \mid (\alpha,\Gamma) \in W_{i+1} \}
\\ &=
\{ \alpha \in \AtC \mid \text{($\alpha \leq p$ and $(\alpha,\Gamma) \in \gI(f) \kc W_i$) or ($\alpha \leq \neg p$ and $\alpha \in \Gamma$}) \}
\\ &=
\{ \alpha \in \AtC \mid \text{$\alpha \leq p$ and $(\alpha,\Gamma) \in \gI(f) \kc W_i$} \}
\cup
(\compl I_\Phi(p) \cap \Gamma)
\\ &=
\{ \alpha \in \AtC \mid \text{$\alpha \leq p$ and $(\alpha,V_i) \in \gI(f)$} \}
\cup
(\compl I_\Phi(p) \cap \Gamma)
\end{align*}
\end{allowdisplaybreaks}%
The last equality above is justified by the following equivalences:
\begin{align*}
&(\alpha,\Gamma) \in \gI(f) \kc W_i \iff
\\
&\text{$\exists Y$.\ $(\alpha,Y) \in \gI(f)$, and $(\beta,\Gamma) \in W_i$ for every $\beta \in Y$}
\\
&\text{$\exists Y$.\ $(\alpha,Y) \in \gI(f)$, and $\beta \in V_i$ for every $\beta \in Y$}
\\
&\text{$\exists Y$.\ $(\alpha,Y) \in \gI(f)$ and $Y \subseteq V_i$},
\end{align*}
which is equivalent to $(\alpha,V_i) \in \gI(f)$. Define the sequence $(U_i)_{i \geq 0}$ by
\begin{align*}
U_0 &=
\{ \alpha \in \AtC \mid
   \text{$\alpha \leq \neg p$ and $\alpha \notin \Gamma$}
\}
\\ &=
\compl I_\Phi(p) \cap \compl \Gamma
\\
U_{i+1} &= U_0 \cup \{ \alpha \in \AtC \mid
  \text{$\alpha \leq p$ and $\forall Y$ with $(\alpha,Y) \in \gI(f)$: $Y \cap U_i \neq \emptyset$}
\}
\end{align*}
Intuitively, the set $U_i$ gives us the atoms from which the demon can force the execution towards the ``error states'' $U_0$ in at most $i$ iterations of the loop.

\begin{claim*}
For every $i \geq 0$, it holds that $V_i = \compl U_i = \AtC \setminus U_i$.
\end{claim*}

Now, we define $U = \bigcup_{i \geq 0} U_i$ and $V = \bigcap_{i \geq 0} V_i$. The above claim implies that $V = \compl U$. Moreover, since $\gI(\wh p f) = \tbigcap_i W_i$, it is easy to see that
\[
  V = \{ \alpha \in \AtC \mid
    (\alpha,\Gamma) \in \gI(\wh p f)
  \}.
\]
Our hypothesis $(\gamma,\Gamma) \in \gI(\wh p f)$ then gives us that $\gamma \in V$.

\begin{claim*}
\label{claim:0closed}
If $\alpha \leq p$ and $\alpha \in V$, then $(\alpha,V)$ is in $\gI(f)$.
\end{claim*}

So, we have the following derivations, where the first one is for an arbitrary $\Phi$-consistent atom $\alpha \in V \cap I_\Phi(p)$:
\begin{allowdisplaybreaks}
\begin{gather*}
\AxiomC{$\begin{array}[b]{ll}
\Phi \vdash (\bigvee V) \land p \to
\\
\hspace{3em}
\bigvee (V \cap I_\Phi(p))
\end{array}$\hspace{-1em}}
\AxiomC{$(\alpha,V)$ in $\gI(f)$}
\RightLabel{(I.H.)}
\UnaryInfC{$\Phi,\Psi \vdash \hoare{\alpha}f{\bigvee V}$}
\AxiomC{$\alpha \in V \cap I_\Phi(p)$}
\RightLabel{($\rJoin$)}
\BinaryInfC{$\Phi,\Psi \vdash \hoare{\bigvee (V \cap I_\Phi(p))}f{\bigvee V}$}
\BinaryInfC{$\Phi,\Psi \vdash \hoare{(\bigvee V) \land p}f{\bigvee V}$}
\RightLabel{($\rLoop$)}
\UnaryInfC{(1) $\Phi,\Psi \vdash \hoare{\bigvee V}{\whileDo p f}{(\bigvee V) \land \neg p}$}
\DisplayProof
\\
\AxiomC{$\gamma \in V$}
\UnaryInfC{$\Phi \vdash \gamma \to r$}
\AxiomC{\vdots}
\AxiomC{$r \triangleq \bigvee V$}
\LeftLabel{(1)}
\BinaryInfC{$\Phi,\Psi \vdash \hoare{r}{\wh p f}{r \land \neg p}$}
\AxiomC{$V \cap \compl I_\Phi(p) = \Gamma \cap \compl I_\Phi(p)$}
\UnaryInfC{$\Phi \vdash r \land \neg p \to \bigvee\Gamma$}
\TernaryInfC{$\Phi,\Psi \vdash \hoare{\gamma}{\whileDo p f}{\bigvee\Gamma}$}
\DisplayProof
\end{gather*}
\end{allowdisplaybreaks}%
The last deduction step is done using the weakening rule.

For angelic choice $f \ang g$, let $(\alpha,X)$ be a pair in $\gI(f \ang g) = \gI(f) \gang \gI(g)$, which is equal to $\gI(f) \cup \gI(g)$. We assume that $(\alpha,X)$ is in $\gI(f)$.
\[
  \AxiomC{$(\alpha,X)$ in $\gI(f)$}
  \RightLabel{(I.H.)}
  \UnaryInfC{$\Phi,\Psi \vdash \hoare{\alpha}f{\bigvee X}$}
  \RightLabel{($\rAng_1$).}
  \UnaryInfC{$\Phi,\Psi \vdash \hoare{\alpha}{f \ang g}{\bigvee X}$}
  \DisplayProof
\]
The case of $(\alpha,X) \in \gI(g)$ is handled analogously.

For demonic choice $f \dem g$, let $(\alpha,X \cup Y)$ be a pair in $\gI(f \dem g) = \gI(f) \gdem \gI(g)$, where $(\alpha,X) \in \gI(f)$ and $(\alpha,Y) \in \gI(g)$. We obtain the derivation:
\[
  \AxiomC{$(\alpha,X)$ in $\gI(f)$}
  \RightLabel{(I.H.)}
  \UnaryInfC{$\Phi,\Psi \vdash \hoare{\alpha}f{\bigvee X}$}
  \AxiomC{$X \subseteq X \cup Y \subseteq \AtC$}
  \UnaryInfC{$\Phi \vdash \bigvee X \to \bigvee (X \cup Y)$}
  \RightLabel{($\rWeak$)}
  \BinaryInfC{$\Phi,\Psi \vdash \hoare{\alpha}f{\bigvee (X \cup Y)}$}
  \DisplayProof
\]
and similarly we also get that
$\Phi,\Psi \vdash
 \hoare{\alpha}g{\bigvee(X \cup Y)}
$.
Finally,
\[
  \AxiomC{$\Phi,\Psi \vdash \hoare{\alpha}f{\bigvee (X \cup Y)}$}
  \AxiomC{$\Phi,\Psi \vdash \hoare{\alpha}g{\bigvee (X \cup Y)}$}
  \RightLabel{($\rDem$)}
  \BinaryInfC{$\Phi,\Psi \vdash \hoare{\alpha}{f \dem g}{\bigvee (X \cup Y)}$}
  \DisplayProof
\]
by the rule for demonic choice, and we are done.
\end{proof}

\begin{corollary}[Completeness]
\label{coro:complete}
Let $\Phi$ be a finite set of tests, and $\Psi$ be a finite set of simple Hoare assertions. For every program $f$, the following are equivalent:
\begin{enumerate}
\item
$\Phi,\Psi \models \hoare{p}f{q}$.
\item
For every $\Phi$-consistent $\alpha \leq p$, the pair $(\alpha,I_\Phi(q))$ is in $\gI(f)$.
\item
$\Phi,\Psi \vdash \hoare{p}f{q}$.
\end{enumerate}
\end{corollary}
\begin{proof}
For the implication (1) $\Imp$ (2), recall that the free interpretation $\gI$ satisfies the hypotheses in $\Phi$ and $\Psi$ (Lemma~\ref{lemma:free}). So, it must be that $\gI$ satisfies $\hoare{p}f{q}$. For a $\Phi$-consistent atom with $\alpha \leq p$, we have that $\gI, \alpha \models p$ and hence $(\alpha,\gI(q))$ is in $\gI(f)$. But $\gI(q) = I_\Phi(q)$, and we thus conclude that $(\alpha,I_\Phi(q)) \in \gI(f)$.

We will prove now the implication (2) $\Imp$ (3). Theorem~\ref{thm:complete} says: $(\alpha,I_\Phi(q)) \in \gI(f)$ implies that
$\Phi,\Psi \vdash
 \hoare{\alpha}f{\bigvee I_\Phi(q)}
$.
So, we have the following deduction
\[
  \AxiomC{$(\alpha,I_\Phi(q))$ in $\gI(f)$}
  \RightLabel{(Thm~\ref{thm:complete})}
  \UnaryInfC{$\Phi,\Psi \vdash \hoare{\alpha}f{\bigvee I_\Phi(q)}$}
  \AxiomC{\hspace{-1em}$\Phi \vdash (\bigvee I_\Phi(q)) \to q$}
  \BinaryInfC{$\Phi,\Psi \vdash \hoare{\alpha}f{q}$}
  \AxiomC{\hspace{-0.5em}$\begin{array}[b]{l}
    \text{for $\alpha \in \AtC$} \\
    \text{with $\alpha \leq p$}
  \end{array}$}
  \RightLabel{($\rJoin$),}
  \BinaryInfC{$\Phi,\Psi \vdash \hoare{\bigvee I_\Phi(p)}f{q}$}
  \DisplayProof
\]
because $I_\Phi(p) = \{ \alpha \in \AtC \mid \alpha \leq p \}$. Finally, notice that $\Phi \vdash p \to \bigvee I_\Phi(p)$ and by the weakening rule we conclude that $\Phi,\Psi \vdash \hoare{p}f{q}$.

The implication (3) $\Imp$ (1) is the soundness result for our Hoare calculus, which we have already proved in Theorem~\ref{thm:sound}.
\end{proof}

Corollary~\ref{coro:complete} gives us a decision procedure for the weak Hoare theory of dual nondeterminism. Given a Hoare implication $\Phi,\Psi \Imp \hoare{p}f{q}$, we simply have to compute the free interpretation $\gI(f) \subseteq \AtC \times \wp\AtC$, which is a finite object. Observe that $\gI(f)$ is of doubly exponential size. We will see later that, with some more work, we can devise a faster algorithm of exponential complexity.

\section{Strong Hoare Theory: Completeness and Complexity}
\label{sec:completeB}

The completeness theorem of \S\ref{sec:completeA} concerns the theory generated by the class of all interpretations, that is, when the atomic programs are allowed to be interpreted as any game function. However, for most realistic applications the atomic actions $a, b, \ldots$ correspond to computational operations (e.g., variable assignments $x := t$, etc.) that involve no angelic nondeterministic choice. This leads us to consider a strictly smaller class of interpretations, and the question is raised of whether this smaller class has the same Hoare theory. This section is devoted to the in-depth study of the theory over this subclass of interpretations. We obtain both an unconditional completeness theorem and a complexity characterization.

\begin{defi}[Validity Over a Class of Interpretations]
We fix a language with atomic tests and atomic actions. Let $\class C$ be a class of interpretations of the atomic symbols (extending to all tests and programs in the usual way). We say that a simple Hoare implication $\Phi,\Psi \Imp \hoare{p}f{q}$ is \emph{valid in $\class C$} (or \emph{$\class C$-valid}) if every interpretation $I$ in $\class C$ satisfies the implication. We then write $\Phi,\Psi \models_{\class C} \hoare{p}f{q}$. The set of all $\class C$-validities is called the \emph{Hoare theory of $\class C$}.

Let $\All$ be the class of all interpretations. Observe that an implication is valid iff it is valid in $\All$. Now, let $\Dem \subseteq \All$ be the strict subclass of interpretations where the atomic actions are interpreted as non-angelic game functions.
\end{defi}

\begin{lemma}[Soundness]
\label{lemma:meetSound}
The rule $(\rMeet)$ of Figure~\ref{fig:meet}, where $a$ is an atomic action, is sound for the class $\Dem$ of interpretations.
\end{lemma}
\begin{proof}
Let $I$ be an interpretation in the class $\Dem$, which means that the game function $I(a): S \gto S$ is non-angelic. Suppose that $I$ satisfies the premises of the rule $(\rMeet)$, and also that it satisfies the hypotheses $\Phi$ and $\Psi$. It follows that $I$ satisfies the assertions $\hoare{p}a{q_1}$ and $\hoare{p}a{q_2}$. We have to show that $I$ also satisfies the assertion $\hoare{p}a{q_1 \land q_2}$. Let $u$ be a state with $u \in I(p)$. Then, we have that $(u,I(q_1)) \in I(a)$ and $(u,I(q_2)) \in I(a)$. Since $I(a)$ is non-angelic, there exists a unique subset $X \subseteq S$ such that $I(a)(u) = \{ Y \subseteq S \mid X \subseteq Y \}$. But $I(q_1)$ and $I(q_2)$ are both in $I(a)(u)$, which means that $X \subseteq I(q_1)$ and $X \subseteq I(q_2)$. We thus obtain that $X \subseteq I(q_1) \cap I(q_2) = I(q_1 \land q_2)$, and therefore $(u,I(q_1 \land q_2)) \in I(a)$. So, $I \models \hoare{p}a{q_1 \land q_2}$, and the proof is complete.
\end{proof}

\begin{figure}[t!]
\normalsize\centering
$\begin{gathered}
\AxiomC{$\Phi,\Psi \vdash \hoare{p}a{q_1}$}
\AxiomC{$\Phi,\Psi \vdash \hoare{p}a{q_2}$}
\RightLabel{($a$-$\rMeet$)}
\BinaryInfC{$\Phi,\Psi \vdash \hoare{p}a{q_1 \land q_2}$}
\DisplayProof
\end{gathered}$
\caption{A rule that is sound when the atomic actions are interpretated as non-angelic game functions. That is, $(\rMeet)$ is sound for the class $\Dem$.}
\label{fig:meet}
\end{figure}

Lemma~\ref{lemma:meetSound} also establishes that the Hoare theory of $\Dem$ is different from the Hoare theory of $\All$. Strictly more implications hold, when we restrict attention to the interpretations of $\Dem$. For example, consider the set of hypotheses $\Psi$, which consists of the two simple assertions $\hoare{p}a{q}$ and $\hoare{p}a{r}$, where $p,q,r$ are distinct atomic tests. Observe that the implication $\Psi \Imp \hoare{p}a{q \land r}$ is valid in $\Dem$ (by Lemma~\ref{lemma:meetSound}), but it is not valid in $\All$ (by virtue of Corollary~\ref{coro:complete}).

\begin{defi}[The Free Non-Angelic Interpretation]
\label{def:freeDem}
Let $\Phi$ be a finite set of tests, and $\Psi$ be a finite set of simple Hoare assertions. For an atomic action $a$, define the nondeterministic interpretation $R_{\Phi\Psi}(a): \AtC \to \wp\AtC$ as
\[
  R_{\Phi\Psi}(a)(\alpha) \triangleq
  \{ \beta \in \AtC \mid
     \text{for every $\hoare{p}a{q} \in \Psi$, $\alpha \leq p$ implies that $\beta \leq q$}
  \}.
\]
We define the \emph{free non-angelic interpretation} $\dI$ (w.r.t.\ $\Phi$ and $\Psi$) to have $\AtC$ as state space, and to interpret the tests as $I_\Phi$ (the free test interpretation w.r.t.\ $\Phi$, see Definition~\ref{def:freeTestI}) does. Moreover, the interpretation $\dI(a): \AtC \gto \AtC$  of the atomic action $a$ is the lifting of $R_{\Phi\Psi}(a)$, that is, it is given by $\dI(a) = \lift R_{\Phi\Psi}(a)$.
\end{defi}

\begin{lemma}
\label{lemma:freeDet}
Let $\Phi$ be a finite set of tests, and $\Psi$ be a finite set of simple Hoare assertions. The free non-angelic interpretation $\dI$ satisfies both $\Phi$ and $\Psi$. \qed
\end{lemma}

Recall that we used the symbol $\vdash$ in \S\ref{sec:hoare} to denote provability in the Hoare-style system of Figure~\ref{fig:angHL}. Now, we will use the symbol $\vdash_d$ to denote provability in the system that extends the calculus of Figure~\ref{fig:angHL} with the additional rule $(\rMeet)$ shown in Figure~\ref{fig:meet}.

\begin{theorem}[Completeness]
\label{thm:completeB}
Let $\Phi$ be a finite set of tests, and $\Psi$ be a finite set of simple Hoare assertions. For every program term $f$ and every $\Phi$-consistent atom $\alpha$,
\[
  \text{$(\alpha,Y) \in \dI(f)$ implies that $\Phi,\Psi \vdash_d \hoare{\alpha}f{\tbigvee Y}$}.
\]
\end{theorem}
\begin{proof}
We will only consider the base case of an atomic program $a$. All the other cases are handled exactly as in Theorem~\ref{thm:complete}, so we omit them. Let $\alpha$ be a $\Phi$-consistent atom. Define
\begin{align*}
X =
R_{\Phi\Psi}(a)(\alpha) &=
\{ \beta \in \AtC \mid
   \text{for all $\hoare{p}a{q} \in \Psi$: $\alpha \leq p$ implies $\beta \leq q$}
\}
\\ &=
I_\Phi \bigl(
  \tbigwedge Q
\bigr),\ 
\text{where $Q = \{ q \mid \text{$\hoare{p}a{q} \in \Psi$ and $\alpha \leq p$} \}$}.
\end{align*}
The claim is that $\Phi,\Psi \vdash_d \hoare{\alpha}a{\bigvee X}$. If the set $Q$ of tests (defined above) is empty, then $\bigwedge Q = \true$ and $X = I_\Phi(\bigwedge Q) = \AtC$. We have the derivation
\[ 
  \AxiomC{$$}
  \RightLabel{$(\rMeet_0)$}
  \UnaryInfC{$\Phi,\Psi \vdash_d \hoare{\alpha}a{\true}$}
  \AxiomC{$\Phi \vdash \true \to \bigvee\AtC$}
  \RightLabel{($\rWeak$).}
  \BinaryInfC{$\Phi,\Psi \vdash_d \hoare{\alpha}a{\true$}}
  \DisplayProof
\]
Now, we can assume that $Q$ is not empty. Using the extra rule $(\rMeet)$ we obtain
\[
  \AxiomC{$\alpha \leq p$}
  \UnaryInfC{$\Phi \vdash \alpha \to p$}
  \AxiomC{$\hoare{p}a{q}$ in $\Psi$}
  \RightLabel{($\rHyp$)}
  \UnaryInfC{$\Phi,\Psi \vdash_d \hoare{p}a{q}$}
  \RightLabel{($\rWeak$)}
  \BinaryInfC{$\Phi,\Psi \vdash_d \hoare{\alpha}a{q}$}
  \AxiomC{$\begin{array}[b]{l}
    \text{for every assertion} \\
    \text{$\hoare{p}a{q}$ in $\Psi$ with} \\
    \text{$\alpha \leq p$}
  \end{array}$}
  \RightLabel{($\rMeet$).}
  \BinaryInfC{(1) $\Phi,\Psi \vdash_d \hoare{\alpha}a{\bigwedge Q}$}
  \DisplayProof
\]
Finally, from $X = I_\Phi(\bigwedge Q)$ we obtain that $\Phi \vdash \bigwedge Q \to \bigvee X$, and using the weakening rule we conclude that
$\Phi,\Psi \vdash_d
 \hoare{\alpha}a{\bigvee X}
$.

Now, let $(\alpha,Y)$ be an arbitrary pair in $\dI$. It follows that $X \subseteq Y$, where $X$ was defined in the previous paragraph. So,
\[
  \AxiomC{$\Phi,\Psi \vdash_d
 \hoare{\alpha}a{\bigvee X}$}
  \AxiomC{$X \subseteq Y \subseteq \AtC$}
  \UnaryInfC{$\Phi \vdash \bigvee X \to \bigvee Y$}
  \RightLabel{($\rWeak$)}
  \BinaryInfC{$\Phi,\Psi \vdash_d \hoare{\alpha}a{\bigvee Y}$}
  \DisplayProof
\]
and the proof is complete.
\end{proof}

\begin{corollary}[Completeness]
\label{coro:completeB}
Let $\Phi$ and $\Psi$ be finite sets of tests and simple Hoare assertions respectively. For every program $f$, the following are equivalent:
\begin{enumerate}
\item
$\Phi,\Psi \models_\Dem \hoare{p}f{q}$.
\item
For every $\Phi$-consistent $\alpha \leq p$, the pair $(\alpha,I_\Phi(q))$ is in $\dI(f)$.
\item
$\Phi,\Psi \vdash_d \hoare{p}f{q}$.
\end{enumerate}
\end{corollary}
\begin{proof}
Similar to the proof of Corollary~\ref{coro:complete}.
\end{proof}

The results so far imply that the Hoare theory of the class $\Dem$, which we also call the \emph{strong Hoare theory} of while game schemes, can be reduced to the weak Hoare theory of the class $\All$ (with an exponential blowup in the size of the instance). Let $\Phi,\Psi \Imp \hoare{p}f{q}$ be an arbitrary Hoare implication. W.l.o.g.\ the axioms in $\Psi$ are of the form $\hoare{\alpha}a{q}$, where $\alpha$ is an atom and $a$ is an atomic action. Now, define $\Psi'$ to be the set of hypotheses that results from $\Psi$ by replacing the axioms $\hoare{\alpha}a{q_i}$ involving $\alpha, a$ by a single axiom $\hoare{\alpha}a{\bigwedge_i q_i}$. The crucial observation is that the interpretation $\dI$ is the same as $I_{\Phi\Psi'}$. Using our two completeness results of Corollary~\ref{coro:complete} and Corollary~\ref{coro:completeB}, it follows that $\Phi,\Psi \vdash_d \hoare{p}f{q}$ iff $\Phi,\Psi' \vdash \hoare{p}f{q}$.

Now, we will investigate the computational complexity of the strong Hoare theory of while game schemes. We prove that this theory is complete for exponential time. In order to obtain the $\EXPTIME$ upper bound, we consider an operational model that corresponds to the free game interpretation. The operational model is a safety game on a finite graph, and we can decide validity by computing the winning regions of the players. The full abstraction result of \S\ref{sec:denotational} says that our denotational semantics coincides in a precise sense to the operational semantics. The lower bound of $\EXPTIME$-hardness is obtained with a reduction from alternating Turing machines with polynomially bounded tapes.

\begin{theorem}[Complexity Upper \& Lower Bound]
\label{thm:exptime}
The strong Hoare theory of while game schemes (the validities over the class $\Dem$) is $\EXPTIME$-complete.
\end{theorem}
\begin{proof}
We first deal with the upper bound. Let $\Phi$ be a finite set of tests, $\Psi$ be a finite set of simple Hoare assertions, and $\hoare{p}f{q}$ be a Hoare assertion. We want to decide whether the simple Hoare implication $\Phi,\Psi \Imp \hoare{p}f{q}$ is valid, equivalently, whether $\Phi,\Psi \vdash \hoare{p}f{q}$. Let $X = I_\Phi(q)$.
According to the completeness result of Corollary~\ref{coro:complete}, we need to check whether $(\alpha,X) \in \dI(f)$ for every $\Phi$-consistent $\alpha \leq p$. By Theorem~\ref{thm:semantics}, this is equivalent to $I_\Phi(p) \times \{f\}$ being contained in the winning region of Player $\exists$ in the safety game $G_{\dI}(f,\compl X)$. Observe in the proof of Theorem~\ref{thm:semantics} that the full abstraction result remains unchanged if in the safety game $G_{\dI}(f,\compl X)$ we only consider the vertices
\begin{align*}
V &=
(\AtC \times C(f)) \cup (\mathcal{X} \times C(f)),
\ \text{where}
\\
\mathcal{X} &= \{
  R_{\Phi\Psi}(a)(\alpha) \mid
  \text{atomic action $a \in C(f)$ and $\alpha \in \AtC$}
\}.
\end{align*}
With this modification, the game $G_{\dI}(f,\compl X)$ is of size exponential in the size of the input: there are exponentially many $\Phi$-consistent atoms, and linearly many terms in $C(f)$ (see Part~\eqref{part:size} of Lemma~\ref{lemma:closureReach}). We can compute the winning regions of $G_{\dI}(f,\compl X)$ in time polynomial in the size of the game. So, overall we need time exponential in the size of the input to decide whether the implication is valid.

We can prove the lower bound by encoding the computations of polynomial-space bounded alternating Turing machines \cite{chandra1981}, since $\EXPTIME = \APSPACE$. An alternating machine consists of the following components: states $Q = Q_\text{and} \cup Q_\text{or}$ (partitioned into and-states \& or-states), input alphabet $\Sigma$, tape alphabet $\Gamma$, blank symbol $\B \in \Gamma$, start state $q_0$, and transition relation
\[
  \Delta \subseteq (Q \times \Gamma) \times (Q \times \Gamma \times \{ -1,0,+1 \}).
\]
We use letters $q, q', \ldots$ to range over states, and $a, b, \ldots$ to range over alphabet symbols. A transition $\lt (q,a),(q',b,d) \rt \in \Delta$ says that if the machine is in state $q$ and is scanning the symbol $a$, then it spawns a new process with its own copy of the tape in which the state is set to $q'$, the symbol $b$ is written over the current position, and the cursor moves by $d$. If $d = -1$ ($d = +1$) the cursor moves one position to the left (right), and if $d = 0$ the cursor stays in the same position. The machine accepts (rejects) if it halts at an and-state (or-state).

The idea is to simulate the alternating machine with a while program scheme that consists of a single while loop. The loop corresponds to the execution loop of the machine, and the body of the loop encodes the transition and process spawning rules (see Figure~\ref{fig:machine}). Without loss of generality we can assume that every computation path halts.

\begin{figure}[t]
\normalsize
\addtolength{\jot}{-3pt}
\centering
$\begin{aligned}
\mathbf{program} \triangleq {}
&\whileDo{(\neg \halt)}{}
\\ &\qquad
\kwIf\ (S_{q_1} \land P_{a_1})\ \kwThen
\\ &\qquad\qquad
\text{take transitions from $(q_1,a_1)$}
\\ &\qquad
\kwElse\ \kwIf\ (S_{q_2} \land P_{a_2})\ \kwThen
\\ &\qquad\qquad
\text{take transitions from $(q_2,a_2)$}
\\[-1.5ex] &\qquad\quad
\vdots
\\ &\qquad
\kwElse\ \kwIf\ (S_{q_m} \land P_{a_m})\ \kwThen
\\ &\qquad\qquad
\text{take transitions from $(q_m,a_m)$}
\\ &\qquad
\kwElse\ \id
\end{aligned}$
\caption{While game scheme that encodes the behavior of an alternating Turing machine.}
\label{fig:machine}
\end{figure}

We introduce atomic tests $P^a_i$ for every tape symbol $a \in \Gamma$ and every position $i$. Intuitively, $P^a_i$ is true when the tape has symbol $a$ at position $i$. The hypotheses
\[\textstyle
  \text{$\bigwedge_i\bigvee_a P^a_i$ \ and \  $\bigwedge_i\bigwedge_{a \neq b} \neg(P^a_i \land P^b_j)$}
\]
say that every position is associated with a unique symbol. We also have atomic tests $C_i$ for every position $i$. The test $C_i$ is true when the cursor is scanning the $i$-th position of the tape. We require that
\[
  \tbigvee_i C_i
  \quad\text{and}\quad
  \tbigwedge_{i \neq j} \neg(C_i \land C_j).
\]
For every state $q \in Q$ of the machine, we introduce an atomic test $S_q$. The test $S_q$ is true when the machine is in state $q$, so we demand:
\[
  \tbigvee_q S_q
  \quad\text{and}\quad
  \tbigwedge_{q \neq q'} \neg(S_q \land S_{q'}).
\]
The machine halts when it is in a state $q$ and the cursor is scanning a symbol $a$ so that the pair $(q,a)$ has no $\Delta$-successor. In this case, we say that the pair $(q,a)$ is a dead-end. So, we define the abbreviations
\begin{align*}
P_a &\triangleq \tbigvee_i (C_i \land P^a_i)
&
\halt &\triangleq
\tbigvee_\text{$q,a$ where $(q,a)$ is dead-end} 
  (S_q \land P_a)
\end{align*}
where $P_a$ says that the currently scanned symbol is $a$, and $\halt$ asserts that the machine can take no transition. Moreover, we define the abbreviations
\begin{align*}
\accept &\triangleq
\halt \land (\tbigvee_{q \in Q_\text{and}} S_q)
&
\reject &\triangleq
\halt \land (\tbigvee_{q \in Q_\text{or}} S_q)
\end{align*}
that describe acceptance and rejection respectively in terms of the atomic tests.


The atomic program $\writeAt a$ writes the symbol $a$ on the tape at the position where the cursor is, and leaves everything else unchanged. So, we take the following hypotheses for it:
\begin{align*}
\hoare{C_i}{&\writeAt a}{P^a_i}
&
\hoare{C_i}{&\writeAt a}{C_i}
\\
\hoare{C_i \land P^b_j}{&\writeAt{a}}{P^b_j},\ \text{for $j \neq i$}
&
\hoare{S_q}{&\writeAt a}{S_q}
\end{align*}
where $i,j$ range over all positions, $b$ ranges over all tape symbols, and $q$ ranges over all machine states. The atomic program $\move d$, where $d \in \{-1,0,1\}$, moves the cursor by $d$. The tape and the machine state remain unchanged.
\begin{align*}
\hoare{C_i}{&\move d}{C_{i+d}}
&
\hoare{P^a_j}{&\move d}{P^a_j}
&
\hoare{S_q}{&\move d}{S_q}
\end{align*}
where $i$ ranges over all positions for which $i+d$ is also a position, $j$ ranges over all positions, $a$ ranges over all tape symbols, and $q$ ranges over all machine states. Finally, we introduce the atomic program $\switch q$, which changes the state of the machine into $q$. The tape and the cursor position remain unchanged.
\begin{align*}
\hoare{\true}{&\switch q}{S_q}
&
\hoare{C_i}{&\switch q}{C_i}
&
\hoare{P^a_i}{&\switch q}{P^a_i}
\end{align*}
where $i$ ranges over all positions, and $a$ ranges over all tape symbols. Suppose that $(q,a)$ is a state-symbol pair that has at least one $\Delta$-successor. If it has exactly one $\Delta$-successor $(q',b,d)$, then we define
\[
  \text{take transitions from $(q,a)$} \triangleq
  \writeAt{b}; \move{d}; \switch{q'}.
\]
If $(q,a)$ has exactly two $\Delta$-successors $(q_1,b_1,d_1)$ and $(q_2,b_2,d_2)$, and $q$ is an and-state, then we define
\begin{align*}
\text{take transitions from $(q,a)$} \triangleq {}
&(\writeAt{b_1}; \move{d_1}; \switch{q_1}) \dem
\\
&(\writeAt{b_2}; \move{d_2}; \switch{q_2}).
\end{align*}
In the case where $(q,a)$ the above $\Delta$-successors but is an or-state, we replace $\dem$ by $\ang$ in the definition. The generalization to more than two $\Delta$-successors is straightforward.

Now, we define the term $\mathbf{program}$ in Figure~\ref{fig:machine} that encodes the execution of the alternating Turing machine. The pairs $(q_1,a_1)$, \ldots, $(q_m,a_m)$ range over the pairs $(q,a)$ that have at least one $\Delta$-successor. For an input string $x_1 x_2 \cdots x_n$, we define the test $\mathit{start}$, which encodes the initial configuration, as
\begin{align*}
\mathit{start} = {}
S_{q_0} \land C_1 \land
\bigl(
  P^{x_1}_1 \land \cdots \land P^{x_n}_n \land
P^{\B}_{n+1} \land \cdots \land P^{\B}_{\pi(n)}
\bigr),
\end{align*}
where $q_0$ is the start state, $1$ is the start position, and $\pi(n)$ is the polynomial that gives the space bound of the machine. Since the space is bounded by a polynomial $\pi(n)$, there are polynomially many positions $i$. So, the size of the program is polynomial in the size of the machine. Finally, the claim is that the machine accepts iff
\[
  \Phi,\Psi \models_\Dem
  \hoare{\mathit{start}}{\mathbf{program}}{\accept},
\]
where $\Phi, \Psi$ are the collections of our assumptions for the atomic tests and the atomic programs respectively.
\end{proof}

It is an immediate corollary of the above theorem that the weak Hoare theory (over the class $\All$) can also be decided in exponential time.

\section{A Complete Hoare-style Calculus for Synthesis}
\label{sec:synthesis}

We introduce in Figure~\ref{fig:synthesis} a Hoare-style calculus which can be used for the deductive synthesis of $\ang$-free programs that satisfy a Hoare specification. It is based on the complete calculus for the Hoare theory of the class $\Dem$, which contains interpretations assigning non-angelic game functions (Definition~\ref{def:gameFunction}) to the atomic programs. This is the calculus of Figure~\ref{fig:angHL} with the extra rule ($a$-$\rMeet$) of Figure~\ref{fig:meet}. The main differences are:
\begin{enumerate}[label=(i)]
\item
The rules $(\rJoin_0)$ and $(\rMeet_0)$ of Figure~\ref{fig:angHL} have been weakened into the rules ($a$-$\rJoin_0$) and ($a$-$\rMeet_0$). This is inconsequential, as we have discussed in Observation~\ref{obs:trivialRules}.
\item
Every conclusion $\hoare{p}f{q}$ is decorated with a $\ang$-free program term $t$, which satisfies the specification $\hoare{p}t{q}$ and implements a winning strategy for the angel in the safety game described by the assertion $\hoare{p}f{q}$.
\end{enumerate}
Another difference that deserves mention is the introduction in Figure~\ref{fig:synthesis} of two new variants $(\rJoin')$ and $(\rJoin'')$ of the standard rule $(\rJoin)$. These rules are not necessary for completeness and they can be omitted without breaking our theorems, but they are useful from a practical viewpoint. The new rules $(\rJoin')$ and $(\rJoin'')$ are sound, and they allow useful shortcuts in the deductive synthesis of $\ang$-free programs.

\begin{figure}[t!]
\centering
$\begin{gathered}
\AxiomC{$\hoare p a q$ in $\Psi$}
\RightLabel{($\rHyp$)}
\UnaryInfC{$\Phi,\Psi \vdash a: \hoare{p}a{q}$}
\DisplayProof
\quad
\AxiomC{\phantom{$pq\Psi$}}
\RightLabel{($\rSkip$)}
\UnaryInfC{$\Phi,\Psi \vdash \id: \hoare{p}\id{p}$}
\DisplayProof
\quad
\AxiomC{\phantom{$pq\Psi$}}
\RightLabel{($\rDvrg$)}
\UnaryInfC{$\Phi,\Psi \vdash \bot: \hoare{p}\bot{q}$}
\DisplayProof
\\[1ex]
\AxiomC{$\begin{aligned}
  \Phi,\Psi &\vdash s: \hoare{p}f{q}
  \\[-0.5ex]
  \Phi,\Psi &\vdash t: \hoare{q}g{r}
\end{aligned}$}
\RightLabel{($\rSeq$)}
\UnaryInfC{$\Phi,\Psi \vdash s;t: \hoare{p}{f;g}{r}$}
\DisplayProof
\qquad
\AxiomC{$\begin{aligned}
  \Phi,\Psi &\vdash s: \hoare{q \land p}f{r}
  \\[-0.5ex]
  \Phi,\Psi &\vdash t: \hoare{q \land \neg p}g{r}
\end{aligned}$}
\RightLabel{($\rCond$)}
\UnaryInfC{$\Phi,\Psi \vdash p[s,t]: \hoare{q}{\ifThenElse p f g}{r}$}
\DisplayProof
\\[1ex]
\AxiomC{$\Phi,\Psi \vdash t: \hoare{r \land p}f{r}$}
\RightLabel{($\rLoop$)}
\UnaryInfC{$\Phi,\Psi \vdash \wh p t: \hoare{r}{\whileDo p f}{r \land \neg p}$}
\DisplayProof
\\[1ex]
\AxiomC{$\Phi,\Psi \vdash t: \hoare{p}{f_i}{q}$}
\RightLabel{($\rAng_i$)}
\UnaryInfC{$\Phi,\Psi \vdash t: \hoare{p}{f_1 \ang f_2}{q}$}
\DisplayProof
\quad
\AxiomC{$\Phi,\Psi \vdash s: \hoare{p}f{q}$}
\AxiomC{$\Phi,\Psi \vdash t: \hoare{p}g{q}$}
\RightLabel{($\rDem$)}
\BinaryInfC{$\Phi,\Psi \vdash s \dem t: \hoare{p}{f \dem g}{q}$}
\DisplayProof
\\[1ex]
\AxiomC{$\Phi \vdash p' \to p$}
\AxiomC{$\Phi,\Psi \vdash t: \hoare{p}f{q}$}
\AxiomC{$\Phi \vdash q \to q'$}
\RightLabel{($\rWeak$)}
\TernaryInfC{$\Phi,\Psi \vdash t: \hoare{p'}f{q'}$}
\DisplayProof
\\[1ex]
\AxiomC{$\Phi,\Psi \vdash t_1: \hoare{p_1}f{q}$}
\AxiomC{$\Phi,\Psi \vdash t_2: \hoare{p_2}f{q}$}
\RightLabel{($\rJoin$)}
\BinaryInfC{$\Phi,\Psi \vdash p_1[t_1,t_2]: \hoare{p_1 \lor p_2}f{q}$}
\DisplayProof
\qquad
\begin{gathered}
\text{($a$-$\rJoin_0$)} \\[-0.5ex]
\Phi,\Psi \vdash a: \hoare{\false}a{q}
\end{gathered}
\\[1ex]
\AxiomC{$\Phi,\Psi \vdash a: \hoare{p}a{q_1}$}
\AxiomC{$\Phi,\Psi \vdash a: \hoare{p}a{q_2}$}
\RightLabel{($a$-$\rMeet$)}
\BinaryInfC{$\Phi,\Psi \vdash a: \hoare{p}a{q_1 \land q_2}$}
\DisplayProof
\qquad
\begin{gathered}
\text{($a$-$\rMeet_0$)} \\[-0.5ex]
\Phi,\Psi \vdash a: \hoare{p}a{\true}
\end{gathered}
\\[1ex]
\AxiomC{$\begin{aligned}
  \Phi,\Psi &\vdash t: \hoare{p_1}f{q}
  \\[-0.5ex]
  \Phi,\Psi &\vdash t: \hoare{p_2}f{q}
\end{aligned}$}
\RightLabel{($\rJoin'$)}
\UnaryInfC{$\Phi,\Psi \vdash t: \hoare{p_1 \lor p_2}f{q}$}
\DisplayProof
\qquad
\AxiomC{$\begin{aligned}
  \Phi,\Psi &\vdash t_1: \hoare{p \land r}f{q}
  \\[-0.5ex]
  \Phi,\Psi &\vdash t_2: \hoare{p \land \neg r}f{q}
\end{aligned}$}
\RightLabel{($\rJoin''$)}
\UnaryInfC{$\Phi,\Psi \vdash r[t_1,t_2]: \hoare{p}f{q}$}
\DisplayProof
\end{gathered}$
\caption{A sound and complete Hoare-style calculus for the synthesis of programs.}
\label{fig:synthesis}
\end{figure}

\begin{theorem}[Soundness]
\label{thm:soundC}
Suppose that a judgment $\Phi,\Psi \vdash t: \hoare{p}f{q}$ is derivable using the Hoare-style calculus of Figure~\ref{fig:synthesis}. The following hold:
\begin{enumerate}
\item
Every game interpretation $I$ in $\Dem$ satisfies the formula $\Phi,\Psi \Imp \hoare{p}f{q}$.
\item
Every nondeterministic interpretation $R$ satisfies $\Phi,\Psi \Imp \hoare{p}t{q}$.
\item
Let $R$ be a nondeterministic interpretation, and $I$ be the game interpretation that lifts $R$ (see Definition~\ref{def:gameI}). Then, $\lift R(t) \subseteq I(f)$.
\end{enumerate}
Part (3) says that $R(t)$ implements $I(f)$, which is denoted $R(t) \simpl I(f)$, when $I$ lifts $R$.
\end{theorem}
\begin{proof}
Part (1) follows from the soundness of the Hoare calculus of Figure~\ref{fig:angHL} (Theorem~\ref{thm:sound}) and from Lemma~\ref{lemma:meetSound} (soundness of the ($a$-$\rMeet$) rule for interpretations in $\Dem$). Part (2) asserts the soundness of a Hoare calculus for nondeterministic while schemes, whose proof can be found in \cite{mamouras2014HL}.
For Part (3), the hypothesis says that $I(a) = \lift R(a)$ for every atomic program $a$, and $I(p) = R(p)$ for every test (see Definition~\ref{def:gameI}). We consider the ``projection'' of the calculus of Figure~\ref{fig:synthesis} to judgments of the form $t: f$, because the rest of the information is irrelevant.
\begin{gather*}
a: a
\qquad
\id: \id
\qquad
\bot: \bot
\qquad
\AxiomC{$s: f$}
\AxiomC{$t: g$}
\BinaryInfC{$s;t: f;g$}
\DisplayProof
\qquad
\AxiomC{$s: f$}
\AxiomC{$t: g$}
\BinaryInfC{$p[s,t]: p[f,g]$}
\DisplayProof
\\[1ex]
\AxiomC{$t: f$}
\UnaryInfC{$\wh p t: \wh p f$}
\DisplayProof
\qquad
\AxiomC{$t: f$}
\UnaryInfC{$t: f \ang g$}
\DisplayProof
\qquad
\AxiomC{$t: g$}
\UnaryInfC{$t: f \ang g$}
\DisplayProof
\qquad
\AxiomC{$s: f$}
\AxiomC{$t: g$}
\BinaryInfC{$s \dem t: f \dem g$}
\DisplayProof
\qquad
\AxiomC{$s: f$}
\AxiomC{$t: f$}
\BinaryInfC{$p[s,t]: f$}
\DisplayProof
\end{gather*}
The claim is that for every derivable judgment $t:f$, we have $R(t) \simpl I(f)$, that is, $R(\phi)$ implements $I(f)$ (see Definition~\ref{def:implementation}). Recall that $R(t) \simpl I(f)$ iff $\lift R(t) \subseteq I(a)$. The proof proceeds by induction on the derivation of $t:f$. It is a straightforward verification, where we make repeated use of Lemma~\ref{lemma:implementation}.
\end{proof}

\begin{theorem}[Completeness]
\label{thm:completeC}
Let $\Phi$ and $\Psi$ be finite sets of tests and simple Hoare assertions respectively, and $f$ be a program s.t.\ $\Phi,\Psi \models_\Dem \hoare{p}f{q}$. Then, there exists a $\ang$-free program $t$ such that $\Phi,\Psi \vdash t: \hoare{p}f{q}$.
\end{theorem}
\begin{proof}
From Corollary~\ref{coro:completeB} we get that $\Phi,\Psi \vdash_d \hoare{p}f{q}$. From Observation~\ref{obs:trivialRules} we know that the rules $(\rJoin_0)$ and $(\rMeet_0)$ can be weakened to ($a$-$\rJoin_0$) and ($a$-$\rMeet_0$) without affecting the provability of the implication $\Phi,\Psi \Imp \hoare{p}f{q}$. We annotate the proof according to the rules of Figure~\ref{fig:synthesis}, and we conclude that $\Phi,\Psi \vdash t:\hoare{p}f{q}$ for some $\ang$-free program $t$. 
\end{proof}

Finally, we will see that solving safety games on finite graphs can be reduced to deciding the $\Dem$-validity of a Hoare implication involving a while game scheme that simulates the safety game. 
Let $G = (V,V_\exists,V_\forall,\to,E)$ be a safety game. For every vertex $u \in V$, introduce an atomic test $p_u$, which asserts that the token is currently on the vertex $u$. We take $\Phi$ to contain the following hypotheses for the atomic tests:
\[
  \tbigvee_{u \in V} p_u
  \qquad\text{and}\qquad
  \text{$\neg (p_u \land p_v)$ for all $u,v \in V$ with $u \neq v$}.
\]
The axioms of $\Phi$ say that the token is on exactly one vertex. So, we can identify the set $\AtC$ of $\Phi$-consistent atoms with the set $\{ p_u \mid u \in V \}$. For every vertex $u \in V$, we introduce an atomic action $u!$, which moves the token to the vertex $u$. So, take $\Psi$ to contain the axioms
\[
  \text{$\hoare{\true}{u!}{p_u}$ for every $u \in V$}.
\]
To emphasize that $\Phi$ and $\Psi$ depend on $G$, let us denote them by $\Phi_G$ and $\Psi_G$ respectively. For an arbitrary vertex $u \in V$, we define the program term
\[
  \text{(take transition from $u$)} \triangleq
  \begin{cases}
    \tbigAng_\text{$v$ with $u \to v$} v!, &\text{if $u \in V_\exists$}
    \\
    \tbigDem_\text{$v$ with $u \to v$} v!, &\text{if $u \in V_\forall$}
    \\
    v!\ (\text{$v$ unique successor of $u$}), &\text{otherwise}
  \end{cases}
\]
Now, we define the while game scheme that describes how the safety game is played:
\begin{align*}
f_G = {}
&\whileDo{(\tbigvee \{ p_u \mid u \in V \setminus E \})}{}
\\ &\qquad
\kwIf\ p_u\ \kwThen\ 
\text{(take transition from $u$)}
\\ &\qquad
\kwElse\ \kwIf\ p_v\ \kwThen\ 
\text{(take transition from $v$)}
\\[-0.5ex] &\qquad
\quad\cdots
\\[-0.5ex] &\qquad
\kwElse\ \kwIf\ p_w\ \kwThen\ 
\text{(take transition from $w$)}
\end{align*}
where $u, v, \ldots, w$ is an enumeration of the non-error vertices. Notice that our encoding implies that a play stops as soon as an error vertex is encountered.

\begin{theorem}[Safety Games]
\label{thm:games}
Let $G = (V,V_\exists,V_\forall,\to,E)$ be a finite safety game. The angel has a winning strategy from $u \in V$ iff $\Phi_G,\Psi_G \vdash \hoare{p_u}{f_G}{\false}$.
\end{theorem}
\begin{proof}
The idea is that Player $\exists$ has a winning strategy from $u$ iff the loop never terminates. The theorem follows immediately from the completeness result of Corollary~\ref{coro:complete} and the operational/denotational correspondence shown in Theorem~\ref{thm:semantics}.
\end{proof}

\section{Example: temperature controller}
\label{sec:example}

\newcommand{\temp}{\mathit{temp}}
\newcommand{\mode}{\mathit{mode}}
\newcommand{\heat}{\mathsf{heat}}
\newcommand{\cool}{\mathsf{cool}}
\newcommand{\off}{\mathsf{off}}
\newcommand{\ok}{\mathit{ok}}
\newcommand{\angel}{\text{angel}}
\newcommand{\demon}{\text{demon}}
\newcommand{\inv}{\mathit{inv}}

We will use our language of while game schemes to encode a toy example of implementing a temperature controller. The idea is that the controller (the angel) can set the heating/cooling system into one of three modes: $\heat$, $\cool$ or $\off$. 
We model this situation with the following program term:
\[
  \angel \triangleq
  (m := \heat) \ang
  (m := \cool) \ang
  (m := \off),
\]
where the variable $m$ stores the current mode. The demon, on the other hand, models the adversarial environment. In particular, he controls the spontaneous temperature changes. We make the simplifying assumption that the temperature can only change by 1 degree Fahrenheit at every time step. Moreover, if the mode is $\heat$ then the temperature cannot decrease, and if the mode is $\cool$ then the temperature cannot increase. We model the behavior of the environment with the term:
\[
  \demon \triangleq
  \begin{aligned}[t]
  &\kwIf\ (m=\heat)\ \kwThen\ (t:=t+1) \dem \id
  \\
  &\kwElse\ \kwIf\ (m=\cool)\ \kwThen\ (t:=t-1) \dem \id
  \\
  &\kwElse\ \kwIf\ (m=\off)\ \kwThen\ (t:=t+1) \dem (t:=t-1) \dem \id,
  \end{aligned}
\]
where the variable $t$ stores the current temperature. The requirement for the temperature controller is that it keeps the temperature within the range $\{ 67, 68, 69 \}$, expressed as
\[
  \ok \triangleq
  (t=67) \lor (t=68) \lor (t=69),
\]
assuming that the initial temperature is 68 degrees Fahrenheit (20 degrees Celsius).

\begin{figure}[t!]
\small
\addtolength{\jot}{-1.5pt}
\centering
\fbox{$\begin{aligned}
&\{ \mathsf{Precondition}: t=68 \}
\\
&\kwWhile\ (t=67)\lor(t=68)\lor(t=69)\ \kwDo
\\
&\qquad
 \begin{aligned}
   &\text{// loop invariant $\inv$:} \\[-0.5ex]
   &\text{//\qquad $(t=67)\lor(t=68)\lor(t=69)$ and} \\[-0.5ex]
   &\text{//\qquad $(t=67) \to (m=\heat)$ and} \\[-0.5ex]
   &\text{//\qquad $(t=69) \to (m=\cool)$}
 \end{aligned}
\\
&\qquad \kwIf\ (m = \heat)\ \kwThen\ (t:=t+1) \dem \id
\\
&\qquad \kwElse\ \kwIf\ (m = \cool)\ \kwThen\ (t:=t-1) \dem \id
\\
&\qquad \kwElse\ \kwIf\ (m = \off)\ \kwThen\ (t:=t+1) \dem (t:=t-1) \dem \id
\\
&\qquad \text{// $(t=67) \lor (t=68) \lor (t=69)$}
\\
&\qquad (m := \heat) \ang (m := \cool) \ang (m := \off)
\\
&\{ \mathsf{Postcondition}: \false \}
\end{aligned}$}
\\[2ex]
$\begin{aligned}[t]
\mathbf{\Phi} : {}
&(t \neq 67) \lor (t \neq 68)
\\
&(t \neq 67) \lor (t \neq 69)
\\
&(t \neq 68) \lor (t \neq 69)
\\
&(m=\heat) \lor (m=\cool) \lor (m=\off)
\\
&(m \neq \heat) \lor (m \neq \cool)
\\
&(m \neq \heat) \lor (m \neq \off)
\\
&(m \neq \cool) \lor (m \neq \off)
\end{aligned}
\quad
\begin{aligned}[t]
\mathbf{\Psi} : {}
&\hoare{t=67}{t := t+1}{t=68}
\\
&\hoare{t=68}{t := t+1}{t=69}
\\
&\hoare{t=69}{t := t+1}{\neg\ok}
\\
&\hoare{m=v}{t := t+1}{m=v},\ \text{for $v=\heat,\cool,\off$}
\\
&\hoare{t=67}{t := t-1}{\neg\ok}
\\
&\hoare{t=68}{t := t-1}{t=67}
\\
&\hoare{t=69}{t := t-1}{t=68}
\\
&\hoare{m=v}{t := t-1}{m=v},\ \text{for $v=\heat,\cool,\off$}
\\
&\hoare{\true}{m := \heat}{m=\heat}
\\
&\hoare{\true}{m := \cool}{m=\cool}
\\
&\hoare{\true}{m := \off}{m=\off}
\\
&\hoare{t=v}{m:=w}{t=v},
\\
&\quad\text{for $v=67,68,69$ and $w=\heat,\cool,\off$}
\\
&\hoare{\neg\ok}{m:=w}{\neg\ok},\ \text{for $w=\heat,\cool,\off$}
\end{aligned}$
\caption{A program modelling the interaction between a temperature controller and the environment, and a Hoare specification for the acceptable temperature range.}
\label{fig:temperature}
\end{figure}

In Figure~\ref{fig:temperature} we see the program that describes the interaction between the controller and the environment (in discrete steps), together with a Hoare specification demanding that the temperature is within the acceptable range. The while loop keeps executing until a violation of the temperature range occurs. In other words, the specification is satisfied when the loop keeps running forever. We assume throughout that we reason under the hypotheses $\Phi$ for atomic tests, and the hypotheses $\Psi$ for atomic actions. The top-level steps of the proof are:
\begin{allowdisplaybreaks}
\begin{align*}
1.\ &
(t=68) \to \inv
&&\text{[$\Phi$, bool]}
\\
2.\ &
\inv \to ((t=67)\land(m=\heat)) \lor (t=68) \lor ((t=69)\land(m=\cool))
&&\text{[bool]}
\\
3.\ &
\hoare{(t=67)\land(m=\heat)}{\demon}{\ok}
&&\text{[use $\Phi,\Psi$]}
\\
4.\ &
\hoare{t=68}{\demon}{\ok}
&&\text{[use $\Phi,\Psi$]}
\\
5.\ &
\hoare{(t=69)\land(m=\cool)}{\demon}{\ok}
&&\text{[use $\Phi,\Psi$]}
\\
6.\ &
\hoare{\inv}{\demon}{\ok}
&&\text{[2, 3, 4, 5]}
\\
7.\ &
\hoare{\ok}{\angel}{\inv}
&&\text{[todo]}
\\
8.\ &
\hoare{\inv}{\demon;\angel}{\inv}
&&\text{[6, 7, $\rSeq$]}
\\
9.\ &
\hoare{\inv}{\whileDo{\ok}{(\demon;\angel)}}{\inv \land \neg \ok}
&&\text{[8, $\rLoop$]}
\\
10.\ &
\inv \land \neg\ok \to \false
&&\text{[bool]}
\\
11.\ &
\hoare{t=68}{\whileDo{\ok}{(\demon;\angel)}}{\false}
&&\text{[1, 9, 10]}
\end{align*}
\end{allowdisplaybreaks}%
It remains to derive the assertion $\hoare{\ok}{\angel}{\inv}$, which concerns the implementation of the controller.
\begin{allowdisplaybreaks}
\begin{align*}
1.\ &
\hoare{t=67}{m:=\heat}{\inv}
&&\text{[use $\Psi$]}
\\
2.\ &
\hoare{t=67}{\angel}{\inv}
&&\text{[1, $\rAng$]}
\\
3.\ &
\hoare{t=69}{m:=\cool}{\inv}
&&\text{[use $\Psi$]}
\\
4.\ &
\hoare{t=69}{\angel}{\inv}
&&\text{[3, $\rAng$]}
\\
5.\ &
\hoare{t=68}{m:=\off}{\inv}
&&\text{[use $\Psi$]}
\\
6.\ &
\hoare{t=68}{\angel}{\inv}
&&\text{[5, $\rAng$]}
\\
7.\ &
\hoare{(t=69) \lor (t=68)}{\angel}{\inv}
&&\text{[4, 6, $\rJoin$]}
\\
8.\ &
\hoare{(t=67) \lor (t=69) \lor (t=68)}{\angel}{\inv}
&&\text{[2, 7, $\rJoin$]}
\\
9.\ &
\hoare{\ok}{\angel}{\inv}
&&\text{[8, bool]}
\end{align*}
\end{allowdisplaybreaks}%
If we annotate the above proof with the angelic strategies according to the synthesis calculus of Figure~\ref{fig:synthesis}, then the implementation for the controller becomes:
\begin{align*}
\text{controller} \triangleq {}
&\kwIf\ (t=67)\ \kwThen\ m:=\heat
\\
&\kwElse\ \kwIf\ (t=69)\ \kwThen\ m:=\cool
\\
&\kwElse\ m:=\off.
\end{align*}
We have thus established deductively that there exists an implementation satisfying the specification, and we have obtain a $\ang$-free program that witnesses this fact.

\section{Related Work}
\label{sec:related}

The present paper is inspired from and builds upon the closely related line of work on the propositional fragment of Hoare logic, called \emph{Propositional Hoare Logic} or PHL \cite{kozen1999PHL, kozen2000PHL, cohen2000PHL, kozen2001PHL, tiuryn2002PHL}. In \cite{mamouras2014HL} and \cite{mamouras2016mrsHL}, a propositional variant of Hoare logic for mutually recursive programs is investigated. The present work differs from all this previous work in considering the combination of angelic and demonic nondeterminism, which presents significant new challenges for obtaining completeness and decision procedures.

The other line of work that largely motivated our investigations here is an extension of Propositional Dynamic Logic (PDL) \cite{pratt1976DL, fischer1977pdl, fischer1979pdl}, called \emph{Game Logic} \cite{parikh1983PGL, parikh1985GL, pauly2003game}. This formalism was introduced more than 30 years ago in \cite{parikh1983PGL}, but there are still no completeness results for full Game Logic. We stress that the theory studied here is \emph{not} a fragment of Game Logic. Even though hypotheses-free Hoare assertions $\hoare{p}f{q}$ can be encoded in Dynamic Logic as partial correctness formulas $p \to [f]q$, there is no direct mechanism for encoding the hypotheses of an implication $\Phi,\Psi \Imp \hoare{p}f{q}$ (which would correspond to some kind of global consequence relation in Dynamic Logic).

We have already discussed in the introduction that there have been proposals of semantic models with the explicit purpose of describing the interaction between angelic and demonic choices in programs: monotonic predicate transformers \cite{back1998, dijkstra1975, morgan1998} and up-closed multirelations \cite{rewitzky2003, martin2004, martin2007, martin2013}. We should note that the latter model of multirelations (relations from the state space $S$ to $\wp S$ or, equivalently, functions $S \to \wp\wp S$) had appeared much earlier in the context of modal logic under the name of \emph{neighborhood semantics} or \emph{Scott-Montague semantics}. See \cite{chellas1980ML} for a textbook presentation of this general semantics (called \emph{minimal models} in \cite{chellas1980ML}), which is useful for analyzing non-normal modal logics. These previous works study semantic objects that are related to our game functions. However, our definition of the algebra of game functions (in particular, the definition of while loops in terms of greatest fixpoints) has not been studied before. Moreover, the precise correspondence between safety games and game functions is novel.

There is an enormous amount of work on logics for the strategic interaction between agents, such as Coalition Logic, Alternating-time Temporal Logic, Strategy Logic, and many more. These logics are mostly inspired from modal and temporal logic \cite{blackburn2001ML}, and they are typically used for reasoning about strategic ability, cooperation, agent knowledge, and so on. The recent books \cite{benthem2014LG} and \cite{benthem2015SR} contain broad surveys of the area. We know of no previous proposal, however, that offers a succinct language for describing safety games and (unconditionally) complete systems for reasoning about safety compositionally.

\emph{Coalition Logic (CL)} \cite{pauly2002CL} is a multi-agent formalism that studies cooperation modalities $[C]$, where $C$ is a subset of a set $N$ of agents/players. A formula $[C]\phi$ is read as follows: ``the agents $C$ can cooperate in order to guarantee outcome $\phi$''. This language is sufficient for describing only very simple multi-player games consisting of finitely many steps, and it lacks a treatment of iteration.

The language of \emph{Alternating-time Temporal Logic (ATL)} \cite{alur1997ATL, alur2002ATL} includes modalities of the form $\llt C \rrt$, where C is a subset of agents. The meaning of a formula $\llt C \rrt \phi$ is given w.r.t.\ a fixed multi-player game and it says that: ``the agents $C$ have a joint strategy so that for every joint strategy of the remaining agents, the computation induced by these strategies satisfies the linear temporal property $\phi$''. For a fixed game, the language of ATL is sufficient for describing safety properties. ATL cannot be used, however, for the compositional description and specification of games. An ATL formula describes a global property of the entire game, where the game is fixed a priori.

\emph{Strategy Logic (SL)} \cite{chatterjee2010SL} is a very powerful extension of ATL that allows explicit quantification over the strategies of the players, instead of treating the strategies implicitly using modalities. By making strategy quantification a primitive of the language, SL can describe interesting notions of non-zero-sum games such as Nash equilibria. Similarly to ATL, SL is interpreted over a single fixed game graph. Thus, the language of SL does not offer syntax for the compositional description and analysis of complex game graphs from simpler ones.

The work of Moggi on monads and computational effects \cite{moggi1991monad}, where concepts from category theory are used to structure the denotational semantics of programs, has inspired work on program logics that are parameterized w.r.t.\ a monad encapsulating the computational effects (e.g., nontermination, probabilities, nondeterminism, and so on) of the programs. Neighborhood models and related models of dual nondeterminism have been shown to give rise to monads. A generic monadic framework for weakest precondition semantics is studied in \cite{hasuo2015WP}, and a relatively complete monadic Hoare logic is proposed in \cite{goncharov2013HL}. As far as we know, none of the works in this line of research provides an operationally justified semantics for dual nondeterminism nor an unconditional completeness result.

\section{Discussion \& Conclusion}
\label{sec:conclusion}

We have considered here the weak (over the class $\All$) and the strong (over the subclass $\Dem$) Hoare theories of dual nondeterminism, and we have obtained sound and unconditionally complete Hoare-style calculi for both of them. We have also shown that both theories can be decided in exponential time, and that the strong Hoare theory is EXPTIME-hard. Finally, we have extended our proof system so that it constructs program terms for the strategies of the angel, thus obtaining a sound and complete calculus for synthesis.

To the best of our knowledge, the present results are the first completeness theorems for logics of while programs that support dual nondeterminism. Handling the case of iteration in the presence of both angelic and demonic nondeterminism requires a careful treatment, since we generally need transfinitely many iterations for the loop approximants. In order to gain confidence that the employed semantics is indeed meaningful, we have shown that it agrees exactly with the intended operational model (based on safety games).

There is still much progress to be made in the problem of axiomatizing Game Logic \cite{parikh1983PGL} or a reasonable variation of it (possibly using a restricted class of models and a different syntax for programs). It also remains an interesting challenge to give equational axiomatizations for dual nondeterminism and iteration in the style of Kleene algebra \cite{kozen1994KA} and Kleene algebra with tests \cite{kozen1997KAT}. For practical applications such equational theories would need to accommodate additional hypotheses for the domain of computation \cite{mamouras2014KA, grathwohl2014KATB, mamouras2015phd}, similarly to the use of hypotheses $\Phi$ and $\Psi$ in our calculi. We hope that the present work will inspire progress for the aforementioned and other related open problems.

\section*{Acknowledgement}

The author would like to thank the anonymous referees for their very helpful comments.

\bibliographystyle{alpha}
\bibliography{dualHL-biblio}

\begin{thebibliography}{Mam15b}

\bibitem[AHK97]{alur1997ATL}
Rajeev Alur, Thomas~A. Henzinger, and Orna Kupferman.
\newblock Alternating-time temporal logic.
\newblock In {\em Proceedings of the 38th Annual Symposium on Foundations of
  Computer Science (FOCS '97)}, pages 100--109, 1997.

\bibitem[AHK02]{alur2002ATL}
Rajeev Alur, Thomas~A. Henzinger, and Orna Kupferman.
\newblock Alternating-time temporal logic.
\newblock {\em Journal of the ACM}, 49(5):672--713, 2002.

\bibitem[Apt81]{apt1981HL}
Krzysztof~R. Apt.
\newblock Ten years of {H}oare's logic: A survey -- {Part I}.
\newblock {\em ACM Transactions on Programming Languages and Systems (TOPLAS)},
  3(4):431--483, 1981.

\bibitem[Apt83]{apt1983HL}
Krzysztof~R. Apt.
\newblock Ten years of {H}oare's logic: A survey -- {Part II}: Nondeterminism.
\newblock {\em Theoretical Computer Science}, 28(1):83--109, 1983.

\bibitem[BdRV01]{blackburn2001ML}
Patrick Blackburn, Maarten de~Rijke, and Yde Venema.
\newblock {\em Modal Logic}, volume~53 of {\em Cambridge Tracts in Theoretical
  Computer Science}.
\newblock Cambridge University Press, 2001.

\bibitem[BvW90]{back1990dual}
Ralph-Johan~R. Back and Joakim von Wright.
\newblock Duality in specification languages: A lattice-theoretical approach.
\newblock {\em Acta Informatica}, 27(7):583--625, 1990.

\bibitem[BvW92]{back1992ADM}
Ralph-Johan~R. Back and Joakim von Wright.
\newblock Combining angels, demons and miracles in program specifications.
\newblock {\em Theoretical Computer Science}, 100(2):365--383, 1992.

\bibitem[BW98]{back1998}
Ralph-Johan Back and Joakim Wright.
\newblock {\em Refinement Calculus: A Systematic Introduction}.
\newblock Springer Heidelberg, 1998.

\bibitem[Che80]{chellas1980ML}
Brian~F. Chellas.
\newblock {\em Modal Logic: An Introduction}.
\newblock Cambridge University Press, 1980.

\bibitem[CHP10]{chatterjee2010SL}
Krishnendu Chatterjee, Thomas~A. Henzinger, and Nir Piterman.
\newblock Strategy logic.
\newblock {\em Information and Computation}, 208(6):677--693, 2010.

\bibitem[CK00]{cohen2000PHL}
Ernie Cohen and Dexter Kozen.
\newblock A note on the complexity of propositional {Hoare} logic.
\newblock {\em ACM Transactions on Computational Logic}, 1(1):171--174, 2000.

\bibitem[CKS81]{chandra1981}
Ashok~K. Chandra, Dexter~C. Kozen, and Larry~J. Stockmeyer.
\newblock Alternation.
\newblock {\em Journal of the Association for Computing Machinery},
  28(1):114--133, 1981.

\bibitem[Coo78]{cook1978}
Stephen~A. Cook.
\newblock Soundness and completeness of an axiom system for program
  verification.
\newblock {\em SIAM Journal on Computing}, 7(1):70--90, 1978.

\bibitem[CvW03]{celiku2003}
Orieta Celiku and Joakim von Wright.
\newblock Implementing angelic nondeterminism.
\newblock In {\em Tenth Asia-Pacific Software Engineering Conference}, pages
  176--185, 2003.

\bibitem[Dij75]{dijkstra1975}
Edsger~W. Dijkstra.
\newblock Guarded commands, nondeterminacy and formal derivation of programs.
\newblock {\em Communications of the ACM}, 18(8):453--457, 1975.

\bibitem[FL77]{fischer1977pdl}
Michael~J. Fischer and Richard~E. Ladner.
\newblock Propositional modal logic of programs.
\newblock In {\em Proceedings of the Ninth Annual ACM Symposium on Theory of
  Computing (STOC '77)}, pages 286--294, 1977.

\bibitem[FL79]{fischer1979pdl}
Michael~J. Fischer and Richard~E. Ladner.
\newblock Propositional dynamic logic of regular programs.
\newblock {\em Journal of Computer and System Sciences}, 18(2):194--211, 1979.

\bibitem[Flo67]{floyd1967}
Robert~W. Floyd.
\newblock Assigning meanings to programs.
\newblock In {\em Mathematical Aspects of Computer Science, Proceedings of AMS
  Symposium in Applied Mathematics}, volume~19, pages 19--32, 1967.

\bibitem[GKM14]{grathwohl2014KATB}
Niels Bj{\o}rn~Bugge Grathwohl, Dexter Kozen, and Konstantinos Mamouras.
\newblock {KAT + B!}
\newblock In {\em Proceedings of the Joint Meeting of the Twenty-Third EACSL
  Annual Conference on Computer Science Logic (CSL) and the Twenty-Ninth Annual
  ACM/IEEE Symposium on Logic in Computer Science (LICS)}, CSL-LICS '14, pages
  44:1--44:10, 2014.

\bibitem[GL73]{garland1973PS}
Stephen~J. Garland and David~C. Luckham.
\newblock Program schemes, recursion schemes, and formal languages.
\newblock {\em Journal of Computer and System Sciences}, 7(2):119--160, 1973.

\bibitem[GS13]{goncharov2013HL}
Sergey Goncharov and Lutz Schr{\"o}der.
\newblock A relatively complete generic {H}oare logic for order-enriched
  effects.
\newblock In {\em Proceedings of the 28th Annual IEEE/ACM Symposium on Logic in
  Computer Science (LICS '13)}, pages 273--282, 2013.

\bibitem[Has15]{hasuo2015WP}
Ichiro Hasuo.
\newblock Generic weakest precondition semantics from monads enriched with
  order.
\newblock {\em Theoretical Computer Science}, 604:2--29, 2015.

\bibitem[Hoa69]{hoare1969}
C.~A.~R. Hoare.
\newblock An axiomatic basis for computer programming.
\newblock {\em Communications of the ACM}, 12(10):576--580,583, 1969.

\bibitem[KM14]{mamouras2014KA}
Dexter Kozen and Konstantinos Mamouras.
\newblock Kleene algebra with equations.
\newblock In {\em Proceedings of the 41st International Colloquium on Automata,
  Languages and Programming (ICALP '14)}, pages 280--292, 2014.

\bibitem[Koz94]{kozen1994KA}
Dexter Kozen.
\newblock A completeness theorem for {K}leene algebras and the algebra of
  regular events.
\newblock {\em Information and Computation}, 110(2):366--390, 1994.

\bibitem[Koz97]{kozen1997KAT}
Dexter Kozen.
\newblock Kleene algebra with tests.
\newblock {\em Transactions on Programming Languages and Systems (TOPLAS)},
  19(3):427--443, 1997.

\bibitem[Koz99]{kozen1999PHL}
Dexter Kozen.
\newblock On {Hoare} logic and {Kleene} algebra with tests.
\newblock In {\em Proceedings of the 14th Symposium on Logic in Computer
  Science (LICS '99)}, pages 167--172, 1999.

\bibitem[Koz00]{kozen2000PHL}
Dexter Kozen.
\newblock On {Hoare} logic and {Kleene} algebra with tests.
\newblock {\em ACM Transactions on Computational Logic}, 1(1):60--76, 2000.

\bibitem[KT01]{kozen2001PHL}
Dexter Kozen and Jerzy Tiuryn.
\newblock On the completeness of propositional {Hoare} logic.
\newblock {\em Information Sciences}, 139(3–-4):187--195, 2001.

\bibitem[LPP70]{luckham1970FCP}
David~C. Luckham, David M.~R. Park, and Michael~S. Paterson.
\newblock On formalised computer programs.
\newblock {\em Journal of Computer and System Sciences}, 4(3):220--249, 1970.

\bibitem[Mam14]{mamouras2014HL}
Konstantinos Mamouras.
\newblock On the {H}oare theory of monadic recursion schemes.
\newblock In {\em Proceedings of the Joint Meeting of the Twenty-Third EACSL
  Annual Conference on Computer Science Logic (CSL) and the Twenty-Ninth Annual
  ACM/IEEE Symposium on Logic in Computer Science (LICS)}, CSL-LICS '14, pages
  69:1--69:10, 2014.

\bibitem[Mam15a]{mamouras2015phd}
Konstantinos Mamouras.
\newblock {\em Extensions Of Kleene Algebra For Program Verification}.
\newblock PhD thesis, Cornell University, Ithaca, NY, August 2015.

\bibitem[Mam15b]{mamouras2015angHL}
Konstantinos Mamouras.
\newblock Synthesis of strategies and the {H}oare logic of angelic
  nondeterminism.
\newblock In Andrew Pitts, editor, {\em Proceedings of the 18th International
  Conference on Foundations of Software Science and Computation Structures
  (FOSSACS '15)}, volume 9034 of {\em Lecture Notes in Computer Science}, pages
  25--40. Springer, 2015.

\bibitem[Mam16]{mamouras2016mrsHL}
Konstantinos Mamouras.
\newblock The {H}oare logic of deterministic and nondeterministic monadic
  recursion schemes.
\newblock {\em ACM Transactions on Computational Logic (TOCL)},
  17(2):13:1--13:30, 2016.

\bibitem[MC13]{martin2013}
Clare~E. Martin and Sharon~A. Curtis.
\newblock The algebra of multirelations.
\newblock {\em Mathematical Structures in Computer Science}, 23:635--674, 2013.

\bibitem[MCR04]{martin2004}
Clare~E. Martin, Sharon~A. Curtis, and Ingrid Rewitzky.
\newblock Modelling nondeterminism.
\newblock In {\em Proceedings of the 7th International Conference on the
  Mathematics of Program Construction (MPC '04)}, pages 228--251, 2004.

\bibitem[MCR07]{martin2007}
Clare~E. Martin, Sharon~A. Curtis, and Ingrid Rewitzky.
\newblock Modelling angelic and demonic nondeterminism with multirelations.
\newblock {\em Science of Computer Programming}, 65(2):140--158, 2007.

\bibitem[Mog91]{moggi1991monad}
Eugenio Moggi.
\newblock Notions of computation and monads.
\newblock {\em Information and Computation}, 93(1):55--92, 1991.

\bibitem[Mor98]{morgan1998}
Carroll Morgan.
\newblock {\em Programming From Specifications}.
\newblock Prentice-Hall, 1998.

\bibitem[Par83]{parikh1983PGL}
Rohit Parikh.
\newblock Propositional game logic.
\newblock In {\em Proceedings of the 24th Annual Symposium on Foundations of
  Computer Science (FOCS '83)}, pages 195--200, 1983.

\bibitem[Par85]{parikh1985GL}
Rohit Parikh.
\newblock The logic of games and its applications.
\newblock In Marek Karplnski and Jan van Leeuwen, editors, {\em Topics in the
  Theory of Computation -- Selected Papers of the International Conference on
  ‘Foundations of Computation Theory’, FCT '83}, volume 102 of {\em
  North-Holland Mathematics Studies}, pages 111--139. North-Holland, 1985.

\bibitem[Pat68]{paterson1968PS}
Michael~S. Paterson.
\newblock Program schemata.
\newblock In {\em Machine Intelligence 3}, pages 19--31. Edinburgh University
  Press, 1968.

\bibitem[Pau02]{pauly2002CL}
Marc Pauly.
\newblock A modal logic for coalitional power in games.
\newblock {\em Journal of Logic and Computation}, 12(1):149--166, 2002.

\bibitem[PH70]{paterson1970CS}
Michael~S. Paterson and Carl~E. Hewitt.
\newblock Comparative schematology.
\newblock In Jack~B. Dennis, editor, {\em Record of the Project MAC Conference
  on Concurrent Systems and Parallel Computation}, pages 119--127. ACM, 1970.

\bibitem[PP03]{pauly2003game}
Marc Pauly and Rohit Parikh.
\newblock Game logic --- {A}n overview.
\newblock {\em Studia Logica}, 75(2):165--182, 2003.

\bibitem[Pra76]{pratt1976DL}
Vaughan~R. Pratt.
\newblock Semantical considerations on {Floyd-Hoare} logic.
\newblock In {\em Proceedings of the 17th IEEE Annual Symposium on Foundations
  of Computer Science (FOCS '76)}, pages 109--121, 1976.

\bibitem[Rew03]{rewitzky2003}
Ingrid Rewitzky.
\newblock Binary multirelations.
\newblock In {\em Theory and Applications of Relational Structures as Knowledge
  Instruments}, pages 256--271. Springer, 2003.

\bibitem[Rut64]{rutledge1964IPS}
Joseph~D. Rutledge.
\newblock On {I}anov's program schemata.
\newblock {\em Journal of the ACM}, 11(1):1--9, 1964.

\bibitem[Tho95]{thomas1995}
Wolfgang Thomas.
\newblock On the synthesis of strategies in infinite games.
\newblock In {\em Proceedings of the 12th Annual Symposium on Theoretical
  Aspects of Computer Science (STACS '95)}, pages 1--13, 1995.

\bibitem[Tiu02]{tiuryn2002PHL}
Jerzy Tiuryn.
\newblock {Hoare} logic: From first-order to propositional formalism.
\newblock In {\em Proof and System-Reliability}, pages 323--340. Springer,
  2002.

\bibitem[vB14]{benthem2014LG}
Johan van Benthem.
\newblock {\em Logic in Games}.
\newblock MIT Press, 2014.

\bibitem[vBGV15]{benthem2015SR}
Johan van Benthem, Sujata Gosh, and Rineke Verbrugge, editors.
\newblock {\em Models of Strategic Reasoning: Logics, Games, and Communities}.
\newblock Springer, 2015.

\end{thebibliography}
\vspace{-30 pt}

\end{document}